\documentclass{rspublic}%

\usepackage{amsfonts}
\usepackage{amsmath}
\usepackage{amssymb}
\usepackage{graphicx}%
\newtheorem{definition}[theorem]{Definition}

\providecommand{\U}[1]{\protect\rule{.1in}{.1in}}

\begin{document}

\title{On determinism and well-posedness in multiple time dimensions}
\author[W. Craig and S. Weinstein]{Walter Craig$^{1}$ and Steven Weinstein$^{2}$}

\affiliation{$^1$Department of Mathematics and Statistics, McMaster
University, \\ Hamilton, ON, L8S 4K1, Canada. Email: craig@math.mcmaster.ca.
\\ \hfill . \\
$^2$Perimeter Institute for Theoretical Physics, 31
Caroline St North, \\ Waterloo, ON\ N2L 2Y5, Canada.
Email:sweinstein@perimeterinstitute.ca. \\ and \\
University of Waterloo Dept. of Philosophy, 200 University Ave W, \\
Waterloo, ON\ \ N2L\ 3G1, Canada.
Email: sw@uwaterloo.ca.}

\maketitle

\begin{abstract}
{ultrahyperbolic equation, nonlocal constraint} \textbf{Abstract: } We study
the initial value problem for the wave equation and the ultrahyperbolic
equation for data posed on an initial hypersurface surface of mixed space- and
timelike signature. \ We show that under a nonlocal constraint, the initial
value problem posed on codimension-one hypersurfaces --- the Cauchy problem
--- has global unique solutions in the Sobolev spaces $H^{m}$. Thus it is
well-posed. \ In contrast, we show that the initial value problem on higher
codimension hypersurfaces is ill-posed due to failure of uniqueness, at least
when specifying a finite number of derivatives of the data. \ This failure is
in contrast to a uniqueness result for data given in an arbitrary neighborhood
of such initial hypersurfaces, which Courant deduces from Asgeirsson's mean
value theorem. We give a generalization of Courant's theorem which extends to
a broader class of equations. \ The proofs use Fourier synthesis and the
Holmgren--John uniqueness theorem.
\end{abstract}

\label{firstpage}


\section{Introduction}

The field equation
\[
\triangle u-\partial_{y}^{2}u=0
\]
in Minkowski spacetime is of central physical importance, as it describes the
propagation of many of the physical quantities described by field theories,
including the components of the electromagnetic field in a vacuum. Its
generalization to a theory which has multiple times is an ultrahyperbolic
equation. The study of these equations provides a useful window onto the
mathematical status of physical theories involving multiple times, and perhaps
more importantly, provides insight into the extent to which the ordinary
concepts of causality and determinism survive the transition to multiple time dimensions.

Consideration of theories with multiple times has been relatively rare because
it is widely believed that they are inherently unstable, and thus are not
deterministic in a physically meaningful sense. Certain significant developments in
theoretical physics, notably string theory, require additional dimensions, and
in most work to date\footnote[1]{Exceptions include the work of Tegmark
(1997), Hull (1999), Hull \& Khuri (2000), and Bars (2001).} the signature for
the extra dimensions is spatial, reflecting in part this concern with
instability. Motivated by this, the purpose of the present paper is to
reconsider the questions of stability, uniqueness and determinism of the
initial value problem in the presence of multiple time dimensions. We take the
model field equation in this setting to be the simple generalizations of the
wave equation to multiple times, the ultrahyperbolic equation. We find
that the issue of stability and uniqueness for the Cauchy problem can be
addressed by imposing nonlocal constraints that arise naturally from the field equations.

It may be thought reasonable to go beyond the traditional Cauchy problem, and
give initial data on hypersurfaces of higher codimension.  We show
that under the above constraints one can preserve
stability in this setting, but uniqueness is lost, and thus
determinism. Indeed, one may specify an arbitrary finite number of normal
derivatives of the solution on the higher codimension hypersurface, and insist
upon smooth solutions, yet still fail to achieve uniqueness. In contrast to
this, we conclude with a result that essentially recovers and generalizes a
theorem of Courant, which shows that the values of a solution in an
arbitrarily small neighborhood of the initial hypersurface are sufficient to
determine the solution uniquely. In related and prior work, Woodhouse
(1992) studied the case of two space and two time dimensions with 
initial data on a spacelike hypersurface (thus of codimension 2), 
using the Penrose twistor transform in the real case. He also recovered
the uniqueness result of Courant and its implicit constraints
on well-posed inital data for the Cauchy problem. Our work
provides a rigorous analytic alternative for his solution
method, which is not restricted to this choice of space and
time dimensions. We remark that none of these results rely
upon properties of analyticity of the data or the solution.

To fix our notation, the wave equation in $d_{1}$-many space 
dimensions and one time dimension is
\begin{equation}
\triangle_{x}u-\partial_{y}^{2}u:=%
{\displaystyle\sum\limits_{j=1}^{d_{1}}}
\partial_{x_{j}}^{2}u-\partial_{y}^{2}u=0\text{ .}%
\end{equation}
The standard Cauchy problem is posed on $N=\left\{  (x,y)\in\mathbb{R}%
_{x}^{d_{1}}\times\mathbb{R}_{y}^{1}:y=0\right\}  $, a spacelike codimension
one linear hypersurface, for initial data
\[
u(x,0)=f(x)\text{, }\partial_{y}u(x,0)=g(x)\text{.}%
\]
A nonstandard Cauchy problem is posed for a linear hypersurface of mixed
signature $N=\left\{  (x,y):x_{1}=0\right\}  \subseteq\mathbb{R}_{x}^{d_{1}%
}\times\mathbb{R}_{y}^{1}$, namely
\[
u(0,x^{\prime},y)=f(x^{\prime},y)\text{, }\partial_{x_{1}}u(0,x^{\prime
},y)=g(x^{\prime},y)\text{,}%
\]
where the notation is that $x=(x_{1},x^{\prime})\in\mathbb{R}^{d_{1}}$.
Courant (1962) calls this the non-spacelike Cauchy problem, but to avoid
confusion with the non-characteristic Cauchy problem, we call it a Cauchy
problem of \emph{mixed signature}.

An ultrahyperbolic equation has the form
\begin{equation}
\triangle_{x}u-\triangle_{y}u:=%
{\displaystyle\sum\limits_{j=1}^{d_{1}}}
\partial_{x_{j}}^{2}u-%
{\displaystyle\sum\limits_{j=1}^{d_{2}}}
\partial_{y_{j}}^{2}u=0\text{ ,} \label{ultra}%
\end{equation}
where $x\in\mathbb{R}^{d_{1}}$ are considered to be the spacelike variables
and $y\in\mathbb{R}^{d_{2}}$ are timelike. The Cauchy problem considers
initial data posed on a linear hypersurface of codimension one. Choosing
$y_{1}$ as the direction of evolution, Cauchy data consist of
\[
u(x,0,y^{\prime})=f(x,y^{\prime})\text{, }\partial_{y_{1}}u(x,0,y^{\prime
})=g(x,y^{\prime})
\]
on the hypersurface $N=\left\{  (x,y)\in\mathbb{R}_{x}^{d_{1}}\times
\mathbb{R}_{y}^{d_{2}}:y_{1}=0\right\}  $.

The initial value problem on a \textit{higher} codimension hypersurface $M$
could take various forms. A natural problem from the perspective of theories
with multiple times is to consider the spacelike hypersurface $M =\left\{
(x,y)\in\mathbb{R}_{x}^{d_{1}}\times\mathbb{R}_{y}^{d_{2}}:y=0\right\}  $ of
codimension $d_{2}$. Alternatively, one may consider more general $M =\{
(x,y)\in\mathbb{R}_{x}^{d_{1}}\times\mathbb{R}_{y}^{d_{2}}:x_{p_{1}+1} =
\dots= x_{d_{1}}=0$, $y_{P_{2}+1} = \dots=y_{d_{2}-1} = 0\} $ where $0\leq
p_{1} \leq d_{1}$ and $0 \leq p_{2} \leq d_{2}-1 $. There is in either case a
question as to how much data, and what constraints, are to be considered on
$M$. Some of the options are to (i) give the zeroth and first normal
derivatives of $u$ on $M$, (ii) give some finite number of derivatives of $u$
on $M$ which are compatible with the constraint imposed by the ultrahyperbolic
equation, or (iii) specify infinitely many derivatives of $u$ on $M$. In this
paper we consider the first two of these cases.

An outline of the results of this paper is as follows. In section 2 we use
Fourier methods to show that the Cauchy problem for the ultrahyperbolic
equation \eqref{ultra} is ill-posed in general but well-posed on Sobolev
spaces $H^{m}$ if an explicit nonlocal constraint is imposed upon the Cauchy
data. This applies as well to the wave equation with Cauchy data on a mixed
signature hypersurface.\ In section 3 we consider the initial value problem
for data given on higher codimension hypersurfaces, and we find that solutions
are highly nonunique for the initial value problems of type (i) and (ii)
above, even among $H^{m}$ smooth solutions and with the imposition of the
constraint given in section~2. \ In particular, for theories with multiple
times that can be transformed to the form of equation (\ref{ultra}), data
posed on the hypersurface $M=\left\{  y=0\right\}  $ do not uniquely determine
the solution at any other point in time $y\in\mathbb{R}^{d_{2}}\backslash
\{0\}$. The extension problem for higher numbers of derivatives is treated by
the same method as case (i) of zeroth and first normal derivatives. Regarding
case (iii), in which one specifies infinitely many derivatives on the initial
hypersurface $M$, we do not have an answer. We do show in section 4 that among
smooth solutions, data in an arbitrarily small ellipsoidal neighborhood of a
disk in $M$ uniquely determine the data in the envelope of its light cones.
This is analogous to a result in Courant (1962) that is derived from
Asgeirsson's mean value theorem.

\section{The Cauchy problem}

Let $x\in\mathbb{R}^{d_{1}}$ and $y\in\mathbb{R}^{d_{2}}$ be the Cartesian
coordinates of space-time, denote $y=(y_{1},y^{\prime})$ and consider the
Cauchy problem of evolution in the coordinate $y_{1}$. The Cauchy problem of
mixed signature that we address is posed as
\begin{equation}
\partial_{y_{1}}^{2}u=\triangle_{x}u-\triangle_{y^{\prime}}u\text{ ,}
\label{ultra_cauchy}%
\end{equation}
with Cauchy data $u(x,0,y^{\prime})=u_{0}(x,y^{\prime})$ and $\partial_{y_{1}%
}$ $u(x,0,y^{\prime})=u_{1}(x,y^{\prime})$. The standard Sobolev spaces
$H^{m}$ of functions of the variables $(x,y^{\prime})$ are defined as closures
of $C_{0}^{\infty}(\mathbb{R}^{d_{1}}\times\mathbb{R}^{d_{2}-1})$ with respect
to the norms
\[
\left\Vert f\right\Vert _{m}^{2}=
{\displaystyle\sum\limits_{\left\vert \alpha\right\vert +\left\vert
\beta\right\vert \leq m}}
{\displaystyle\int}
\left\vert \partial_{x}^{\alpha}\partial_{y^{\prime}}^{\beta}f(x,y^{\prime
})\right\vert ^{2}dxdy^{\prime}\text{ .}
\]
Additionally, there is an energy functional, of indefinite sign, that is
associated with equation (\ref{ultra_cauchy}), namely%
\[
E(u):=\frac{1}{2}%
{\displaystyle\iint}
\left\vert \partial_{y_{1}}u\right\vert ^{2}+\left\vert \nabla_{x}u\right\vert
^{2}-\left\vert \nabla_{y^{\prime}}u\right\vert ^{2}dxdy^{\prime}\text{ .}
\]

\begin{theorem}
\label{thm1}
Suppose that the evolution mapping $y_{1}\rightarrow\binom{u(x,y_{1}%
,y^{\prime})}{\partial_{y_{1}}u(x,y_{1},y^{\prime})}$ is in $C^{1}%
(\mathbb{R}_{y_{1}}:H^{1}\times H^{0})$. Then the energy is conserved along a
solution $u(\cdot,y_{1},\cdot)$:%
\[
E(u(\cdot,y_{1},\cdot))=E(u(\cdot,0,\cdot))\text{.}
\]
\end{theorem}

\begin{proof}
Given $\binom{u(x,y_{1},y^{\prime})}{\partial_{y_{1}}u(x,y_{1},y^{\prime})}\in
C^{1}$, the following calculation is justified:%
\[%
\begin{array}
[c]{cl}%
\partial_{y_{1}}E(u) & =%
{\displaystyle\iint}
(\partial_{y_{1}}u\cdot\partial_{y_{1}}^{2}u+\nabla_{x}u\cdot\nabla
_{x}\partial_{y_{1}}u-\nabla_{y^{\prime}}u\cdot\nabla_{y^{\prime}}%
\partial_{y_{1}}u)dxdy^{\prime}\\
& =%
{\displaystyle\iint}
\partial_{y_{1}}u(\partial_{y_{1}}^{2}u+\triangle_{x}u+\triangle_{y^{\prime}%
}u)dxdy^{\prime}\\
& =0\text{ .}%
\end{array}
\]
\end{proof}

The key issue is that the Cauchy problem above for equation
(\ref{ultra_cauchy}) is ill-posed for $d_{2}\geq2$ and solutions are \emph{not} in
general in $C^{1}(\mathbb{R}_{y_{1}}:H^{1}\times H^{0})$. \ The energy is
indefinite and in particular not bounded below, hence it does not in general
define an energy norm with which to control the Sobolev norms of solutions of
the evolution equations.

To move to the next level of analysis, we give a Fourier synthesis of the
evolution operator for the Cauchy problem of mixed signature. \ Given
$\binom{u_{0}}{u_{1}}\in H^{m+1}\times H^{m}$, consider the Fourier space
variables $(x,y^{\prime})\rightarrow(\xi,\eta^{\prime})$ and define the
Fourier transform in the standard way,
\[
\binom{\hat{u}_{0}(\xi,\eta^{\prime})}{\hat{u}_{1}(\xi,\eta^{\prime})}%
=\frac{1}{\sqrt{2\pi}^{d}}%
{\displaystyle\iint}
e^{-i\xi\cdot x}e^{-i\eta^{\prime}\cdot y^{\prime}}\binom{u_{0}(x,y^{\prime}%
)}{u_{1}(x,y^{\prime})}dxdy^{\prime}%
\]
where $d=(d_{1}+d_{2}-1)$. On a formal level equation (\ref{ultra_cauchy})
under Fourier transform will read
\[
\partial_{y_{1}}^{2}\hat{u}=(-\left\vert \xi\right\vert ^{2}+\left\vert
\eta^{\prime}\right\vert ^{2})\hat{u}\text{ ,}%
\]
giving rise to the expression for the propagator, $\exp(y_{1}\sqrt
{\triangle_{x}-\triangle_{y^{\prime}}})$. The solution thus reads%

\[
\hat{u}(\xi,y_{1},\eta^{\prime})=\cos(\sqrt{\left\vert \xi\right\vert
^{2}-\left\vert \eta^{\prime}\right\vert ^{2}}y_{1})\hat{u}_{0}(\xi
,\eta^{\prime})+\frac{\sin(\sqrt{\left\vert \xi\right\vert ^{2}-\left\vert
\eta^{\prime}\right\vert ^{2}}y_{1})}{\sqrt{\left\vert \xi\right\vert
^{2}-\left\vert \eta^{\prime}\right\vert ^{2}}}\hat{u}_{1}(\xi,\eta^{\prime
})\text{ }%
\]
for $\left\vert \eta^{\prime}\right\vert \leq\left\vert \xi\right\vert $,
while%
\[
\hat{u}(\xi,y_{1},\eta^{\prime})=\cosh(\sqrt{\left\vert \eta^{\prime
}\right\vert ^{2}-\left\vert \xi\right\vert ^{2}}y_{1})\hat{u}_{0}(\xi
,\eta^{\prime})+\frac{\sinh(\sqrt{\left\vert \eta^{\prime}\right\vert
^{2}-\left\vert \xi\right\vert ^{2}}y_{1})}{\sqrt{\left\vert \eta^{\prime
}\right\vert ^{2}-\left\vert \xi\right\vert ^{2}}}\hat{u}_{1}(\xi,\eta
^{\prime})\text{ }%
\]
for $\left\vert \xi\right\vert <\left\vert \eta^{\prime}\right\vert $. That
is, the dispersion relation
\begin{equation}
\omega(\xi,\eta^{\prime})=\sqrt{\left\vert \xi\right\vert ^{2}-\left\vert
\eta^{\prime}\right\vert ^{2}} \label{dispersion}%
\end{equation}
holds in the Fourier space region $\left\{  \left\vert \eta^{\prime
}\right\vert \leq\left\vert \xi\right\vert \right\}  $, while in the
complementary region the evolution of a Fourier mode is described by the
Lyapunov exponent
\begin{equation}
\lambda(\xi,\eta^{\prime})=\sqrt{\left\vert \eta^{\prime}\right\vert
^{2}-\left\vert \xi\right\vert ^{2}}. \label{Lyapunov}%
\end{equation}

When the propagator is applied to data $\binom{u_{0}}{u_{1}}$ which is
analytic, this solution exists for at least short time; for analytic data of
exponential type, the solution is global. \ However, it is clear that general
initial data in $H^{m+1}\times H^{m}$ do not even give rise to solutions which
are tempered distributions for any nonzero $y_{1}$.

On the other hand, imposing a constraint on the initial data, the solution
process is well defined. The fact that some constraint is necessary is indeed
evident from the Asgeirsson mean value theorem, and its consequences, as
discussed in Courant (1962). The form of this nonlocal constraint is evident
from the Fourier synthesis, as we shall now see.

Define a phase space $X$ using an energy norm adapted to the propagator of
equation (\ref{ultra_cauchy}). Using the definition of the dispersion relation
\eqref{dispersion} and the Lyapunov exponent \eqref{Lyapunov}, and the
Plancherel identity, set $v=\binom{v_{0}}{v_{1}}$ and
\begin{align*}
\left\Vert v\right\Vert _{X}^{2}:=  &  \iint\limits_{\left\{  \left\vert
\eta^{\prime}\right\vert <\left\vert \xi\right\vert \right\}  }\omega^{2}%
(\xi,\eta^{\prime})\left\vert \hat{v}_{0}(\xi,\eta^{\prime})\right\vert
^{2}d\xi d\eta^{\prime}\\
&  +\iint\limits_{\left\{  \left\vert \xi\right\vert \leq\left\vert
\eta^{\prime}\right\vert \right\}  }\lambda^{2}(\xi,\eta^{\prime})\left\vert
\hat{v}_{0}(\xi,\eta^{\prime})\right\vert ^{2}d\xi d\eta^{\prime}\\
&  +\iint\left\vert \hat{v}_{1}(\xi,\eta^{\prime})\right\vert ^{2}d\xi
d\eta^{\prime}~\text{.}%
\end{align*}
This is a norm, unlike the actual energy associated with the equation
(\ref{ultra_cauchy}), and can be used to control solutions when the propagator
is restricted to the appropriate stable and/or unstable subspaces of $X$.
Define%
\begin{align}
X^{S}  &  =\left\{  v=\binom{v_{0}}{v_{1}}\in X:\ \frac{1}{2}\left(  \hat
{v}_{0}(\xi,\eta^{\prime})+\frac{\hat{v}_{1}(\xi,\eta^{\prime})}{\lambda
(\xi,\eta^{\prime})}\right)  =0\text{ for }\left\vert \xi\right\vert
<\left\vert \eta^{\prime}\right\vert \right\} \\
X^{U}  &  =\left\{  v\in X:\ \frac{1}{2}\left(  \hat{v}_{0}(\xi,\eta^{\prime
})-\frac{\hat{v}_{1}(\xi,\eta^{\prime})}{\lambda(\xi,\eta^{\prime})}\right)
=0\text{ for }\left\vert \xi\right\vert <\left\vert \eta^{\prime}\right\vert
\right\}
\end{align}
and%
\begin{align*}
X^{C}  &  =\left\{  v\in X:\mathrm{supp}\binom{\hat{v}_{0}}{\hat{v}_{1}}%
(\xi,\eta^{\prime})\subseteq\text{ }\left\{  \left\vert \xi\right\vert
>\left\vert \eta^{\prime}\right\vert \right\}  \right\} \\
&  =X^{S}\cap X^{U}\text{.}%
\end{align*}
The subspace $X^{S}$ corresponds to the center stable subspace for evolution
in $y_{1}\in\mathbb{R}^{+}$, the subspace $X^{U}$ corresponds to the center
unstable subspace, and $X^{C}$ is the center subspace. This nomenclature is
supported by the following theorem.

\begin{theorem}
\label{thm2}
For $\binom{u_{0}}{u_{1}}\in X^{S}$, the Cauchy problem of mixed signature
for equation (\ref{ultra_cauchy}) has a unique solution in $X$ for all $y_{1}
\in\mathbb{R}^{+}$. For $\binom{u_{0}}{u_{1}}\in X^{U}$ the problem has a
unique solution for all $y_{1}\in\mathbb{R}^{-}$, and whenever $\binom{u_{0}
}{u_{1}}\in X^{C}$ the solution exists globally in $y_{1}\in\mathbb{R}$. In
each of these cases, the map $y_{1}\rightarrow u(x,y_{1},y^{\prime})$ is $C^{1}$.
\end{theorem}

Denote the propagators by $\Phi^{S},\Phi^{U}$ and $\Phi^{C}$ for data in the
respective subspaces. \ These solutions are continuous with respect to their
Cauchy data taken in the respective subspaces. This is the result of the next theorem.

\begin{theorem}
\label{thm3}
Given two phase space points $u=\binom{u_{0}}{u_{1}},v=\binom{v_{0}}{v_{1}}\in
X^{S}$, then for $y_{1}\geq0,$
\begin{equation}
\left\Vert \Phi_{y_{1}}^{S}(u)-\Phi_{y_{1}}^{S}(v)\right\Vert _{X}^{2}
\leq\left\Vert u-v\right\Vert _{X}^{2} ~. \label{boundS}
\end{equation}
The analogous estimate holds for $u,v\in X^{U}$, for $y_{1}\leq0$:
\begin{equation}
\left\Vert \Phi_{y_{1}}^{U}(u)-\Phi_{y_{1}}^{U}(v)\right\Vert _{X}^{2}
\leq\left\Vert u-v\right\Vert _{X}^{2} ~. \label{boundU}
\end{equation}
For $u\in X^{C}$, $\Phi_{y_{1}}^{C}=\Phi_{y_{1}}^{S}$ for $y_{1}\geq0$ and
$\Phi_{y_{1}}^{C}=\Phi_{y_{1}}^{U}$ for $y_{1}\leq0$, and equality holds in
both (\ref{boundS}) and (\ref{boundU}).
\end{theorem}

\begin{proof}
It suffices in Theorem \ref{thm3} to prove the first statement. In $X^{S}$ the solution
has two components, distinguished by their Fourier support. Consider first
$\binom{u_{0}}{u_{1}}$ such that $\mathrm{supp}(\hat{u}_{0},\hat{u}%
_{1})\subseteq\left\{  \left\vert \xi\right\vert \geq\left\vert \eta^{\prime
}\right\vert \right\}  := R_{1} $, which gives the center component of the
evolution. The propagator is expressed%
\[
\mathcal{F}\Phi_{y_{1}}^{S}\binom{u_{0}}{u_{1}}=\left(
\begin{array}
[c]{cc}%
\cos(\omega y_{1}) & \frac{\sin(\omega y_{1})}{\omega}\\
-\omega\sin(\omega y_{1}) & \cos(\omega y_{1})
\end{array}
\right)  \binom{\hat{u}_{0}}{\hat{u}_{1}}%
\]
where $\omega=\omega(\xi,\eta^{\prime})$
is the dispersion relation \eqref{dispersion}.
Evaluating this in the energy norm,%
\begin{align}
\left\Vert \Phi_{y_{1}}^{S}\binom{u_{0}}{u_{1}}\right\Vert _{X}^{2}  &
=\iint\left\vert \cos(\omega y_{1})\hat{u}_{0}+\frac{\sin(\omega y_{1}%
)}{\omega}\hat{u}_{1}\right\vert ^{2}\omega^{2}\label{E conserve}\\
&  +\left\vert -\omega\sin(\omega y_{1})\hat{u}_{0}+\cos(\omega y_{1})\hat
{u}_{1}\right\vert ^{2}d\xi d\eta^{\prime}\nonumber\\
&  =\iint(\left\vert \hat{u}_{0}\right\vert ^{2}\omega^{2}+\left\vert \hat
{u}_{1}\right\vert ^{2})d\xi d\eta^{\prime}\nonumber\\
&  =\left\Vert \dbinom{u_{0}}{u_{1}}\right\Vert _{X}^{2}\text{ .}\nonumber
\end{align}
The propagator on the complementary space is more sensitive. Let us suppose
that $\mathrm{supp}(\hat{u}_{0},\hat{u}_{1})\subseteq\left\{  \left\vert
\eta^{\prime}\right\vert >\left\vert \xi\right\vert \right\}  $, then
$\lambda(\xi,\eta^{\prime})>0$ and we express the propagator in terms of its
Fourier transform as
\begin{align*}
\mathcal{F}\Phi_{y_{1}}^{S}\binom{u_{0}}{u_{1}}  &  =\left(
\begin{array}
[c]{cc}%
\cosh(\lambda y_{1}) & \frac{\sinh(\lambda y_{1})}{\lambda}\\
\lambda\sinh(\lambda y_{1}) & \cosh(\lambda y_{1})
\end{array}
\right)  \binom{\hat{u}_{0}}{\hat{u}_{1}}\\
&  =\frac{e^{\lambda y_{1}}}{2}\left(
\begin{array}
[c]{cc}%
1 & \frac{1}{\lambda}\\
\lambda & 1
\end{array}
\right)  \binom{\hat{u}_{0}}{\hat{u}_{1}}+\frac{e^{-\lambda y_{1}}}{2}\left(
\begin{array}
[c]{cc}%
1 & -\frac{1}{\lambda}\\
-\lambda & 1
\end{array}
\right)  \binom{\hat{u}_{0}}{\hat{u}_{1}}\text{ .}%
\end{align*}
The subspace $X^{S}$ consists of precisely those data which lie in
the null space of the first term; this is the expression of the constraint
\begin{equation}
\lambda\hat{u}_{0}(\xi,\eta^{\prime})+\hat{u}_{1}(\xi,\eta^{\prime})=0
\label{Constraint}%
\end{equation}
Measuring the remaining term in energy norm, we find
\begin{align*}
\left\Vert \Phi_{y_{1}}^{S}\binom{u_{0}}{u_{1}}\right\Vert _{X}^{2}  &
=\iint\frac{e^{-2\lambda y_{1}}}{4}\left[  \left\vert \hat{u}_{0}-\frac
{\hat{u}_{1}}{\lambda}\right\vert ^{2}\lambda^{2}+\left\vert -\lambda\hat
{u}_{0}+\hat{u}_{1}\right\vert ^{2}\right]  d\xi d\eta^{\prime}\\
&  \leq\iint e^{-2\lambda y_{1}}\left(  \left\vert \hat{u}_{0}\right\vert
^{2}\lambda^{2}+\left\vert \hat{u}_{1}\right\vert ^{2}\right)  d\xi
d\eta^{\prime}\text{ .}%
\end{align*}
Since we consider the propagator $\Phi_{y_{1}}^{S}$ for $y_{1}\geq0$, the
exponent $-\lambda y_{1}$ is negative, and therefore%
\[
\left\Vert \Phi_{y_{1}}^{S}\binom{u_{0}}{u_{1}}\right\Vert _{X}^{2}%
\leq\left\Vert \binom{u_{0}}{u_{1}}\right\Vert _{X}^{2}%
\]
for $u=\dbinom{u_{0}}{u_{1}}\in X^{S}$. For general data in $X^{S}$, one
decomposes it into its components with support in $\left\{  \left\vert
\xi\right\vert \geq\left\vert \eta^{\prime}\right\vert \right\}  $ for which
we use (\ref{E conserve}), and its component supported in $\left\{  \left\vert
\eta^{\prime}\right\vert >\left\vert \xi\right\vert \right\}  $, which in
addition satisfies the constraint (\ref{Constraint}). Therefore on $X^{S}$%
\[
\left\Vert \Phi_{y_{1}}^{S}(u)\right\Vert _{X}^{2}\leq\left\Vert u\right\Vert
_{X}^{2}\text{ .}%
\]
Bounded operators on $X^{S}$ are continuous with respect to $u\in X^{S}$, and
it is easy to see that the solution behaves continuously with respect to
$y_{1}\geq0$ as well. The case for the subspace $X^{U}$ is proved by the same
arguments, after reversing time $y_{1}\longrightarrow-y_{1}$. This proves
Theorems \ref{thm2} and \ref{thm3}. We remark that on the center subspace $X^{C}$, which yields global solutions, both
constraints are imposed%
\begin{equation}
\lambda\hat{u}_{0}(\xi,\eta^{\prime})\pm\hat{u}_{1}(\xi,\eta^{\prime})=0\text{
,}%
\end{equation}
implying that $\hat{u}_{0}(\xi,\eta^{\prime})=0=\hat{u}_{1}(\xi,\eta^{\prime
})$ for all $\left\{  \left\vert \xi\right\vert \leq\left\vert \eta^{\prime
}\right\vert \right\}  $.
\end{proof}

The proof extends to the initial value problem posed in higher energy spaces,
defined by%
\[
\left\Vert v\right\Vert _{X^{m}}^{2}:=\sum\limits_{\left\vert \alpha
\right\vert +\left\vert \beta\right\vert \leq m}\left\Vert \partial
_{x}^{\alpha}\partial_{y^{\prime}}^{\beta}v\right\Vert _{X}^{2}\text{ .}
\]
We then have

\begin{corollary}
The higher energy space $X^{m}$ decomposes into three subspaces, $X^{m,S},$
$X^{m,U}$ and $X^{m,C}=X^{m,S}\cap X^{m,U}$ such that for $u,v\in X^{m,S}$ and
$y_{1}\geq0$%
\[
\left\Vert \Phi_{y_{1}}^{S}(u)-\Phi_{y_{1}}^{S}(v)\right\Vert _{X^{m}}^{2}%
\leq\left\Vert u-v\right\Vert _{X^{m}}^{2}\text{ ,}
\]
while for $y_{1}\leq0$ and $u,v\in X^{m,U}$,%
\[
\left\Vert \Phi_{y_{1}}^{U}(u)-\Phi_{y_{1}}^{U}(v)\right\Vert _{X^{m}}^{2}%
\leq\left\Vert u-v\right\Vert _{X^{m}}^{2}\text{ .}
\]
For $u,v\in X^{m,C}$ both estimates hold, and a global solution exists which
has properties of higher Sobolev regularity. When $m>((d_{1}+(d_{2}-1))/2)+2$
then such solutions are known to be classical $C^{2}$ solutions by the Sobolev
embedding theorem.
\end{corollary}

It is natural to estimate solutions with respect to the energy norm; indeed,
it \emph{is }the energy when restricted to the center subspace $X^{C}$. Thus
the problem is well-posed in the following sense:\ data in $X^{S}$
continuously propagates to all $y_{1}\in\mathbb{R}^{+}$, data in $X^{U}$
continuously propagates to all times $y_{1}\in\mathbb{R}^{-}$, and data in
$X^{C}$, which constitute an infinite-dimensional Hilbert space, are defined
globally in time. In the case of the ordinary wave equation ($d_{1}=1$),
solutions in $X^{C}$ correspond to the full energy space $H^{1} \times L^{2}$.

\section{The initial value problem in higher codimension}

In the presence of multiple time dimensions, spacelike hypersurfaces are
necessarily of higher codimension. Therefore one might consider the initial
value problem with data posed on a hypersurface of codimension greater than or
equal to two. Such problems are generally ill-posed. Indeed, we show that
solutions can be singular for standard classes of data. Moreover, even
imposing the constraint discussed in section 2, which is the requirement of
global existence, smooth solutions are highly non-unique. The purpose of this
section is to study the extension problem of data posed on a non-degenerate
higher codimension hypersurface $M$ to Cauchy data on a codimension one
hypersurface $N$. There is a lot of freedom in choosing this extension, even
under the constraint equation (\ref{Constraint}) on the resulting Cauchy data.
Other extensions can be chosen to fail to satisfy the constraint. Thus the
initial value problem fails to be well-posed in several ways: resulting
solutions may be singular, or they may be selected to satisfy the constraint
and be regular for all $y_{1}\in\mathbb{R}$, however they will not be unique.

\subsection{Codimension 2 to codimension 1 in $\mathbb{R}^{3}$}

Our analysis is illustrated in the example case of $M=\{y_{1}=y_{2}=0\}$ and
$N=\{y_{1}=0\}$ subspaces of $\mathbb{R}^{3}$. We suppose that initial data
for a solution $u(x,y)$ is given on $M$ in the form
\[
w(x_{1})=(w_{0}(x_{1}),w_{10}(x_{1}),w_{01}(x_{1}))
\]
where $w_{0}(x_{1})=u(x_{1},0)$, $w_{10}(x_{1})=\partial_{y_{1}}u(x_{1},0)$
and $w_{01}(x_{1})=\partial_{y_{2}}u(x_{1},0)$, corresponding to the values of
the solution and its normal derivatives on $M$. The object is to extend
$w(x_{1})$ to Cauchy data $(u_{0}(x_{1},y_{2}),u_{1}(x_{1},y_{2})$ on $N$
which satisfies%
\begin{align*}
u_{0}(x_{1},0)  &  =w_{0}(x_{1})\\
u_{1}(x_{1},0)  &  =w_{10}(x_{1})
\end{align*}
and the compatibility condition%
\[
\partial_{y_{2}}u_{0}(x_{1},0)=w_{01}(x_{1})\text{.}%
\]
We give extensions which satisfy the constraint (\ref{Constraint}), therefore
giving rise to global solutions in $y_{1}\in\mathbb{R}$. Such extensions are
nonunique. Additionally, there are extensions which fail to satisfy the
constraint, lying in $X^{S}\backslash X^{C}$ or $X^{U}\backslash X^{C}$ or neither.

\begin{definition}
The extension operator is given by%
\[
E(w)(x_{1},y_{2}):=\frac{1}{2\pi}\iint e^{i(\xi x_{1}+\eta^{\prime}y_{2})}
\hat{w}(\xi)\chi(\xi,\eta^{\prime})d\eta^{\prime}d\xi
\]
where the kernel function $\chi(\xi,\eta^{\prime})$ is chosen such that for
all $\xi$,%
\[
\frac{1}{\sqrt{2\pi}}\int\chi(\xi,\eta^{\prime})d\eta^{\prime}=1\text{.}
\]
In order to satisfy the constraint, we ask additionally that
$\mathrm{supp(}
\chi(\xi,\eta^{\prime}))\subseteq\{\left\vert \eta^{\prime}\right\vert
<\left\vert \xi\right\vert \}$. A reasonable choice is to take%
\[
\chi(\xi,\eta^{\prime})=\psi(\eta^{\prime}/\left\vert \xi\right\vert )\frac
{1}{\left\vert \xi\right\vert }\text{,}%
\]
for $\psi(\theta)\in C_{0}^{\infty}([-1,1])$, $\psi(\theta)\geq0$ even, and
\begin{equation}
\frac{1}{\sqrt{2\pi}}\int_{-1}^{1}\psi(\theta)d\theta=1. \label{Normalization}%
\end{equation}
\end{definition}

\begin{theorem}
\label{thm4}
The extension operator $E$ is a bounded operator on the following space of
functions:
\begin{eqnarray*}
  &&  E:\dot{H}^{-1/2}(M)\rightarrow L^{2}(N)  \\
  &&  \quad \ (\dot{H}^{-1/2}\cap H^{m-1/2})(M)\rightarrow H^{m}(N) ~.
\end{eqnarray*}
In addition, when $w\in\dot{H}^{-3/2}(M)$ then $y_{2}E(w)\in L^{2}(N)$ and
furthermore
\begin{eqnarray*}
  && y_{2}E:\dot{H}^{m-3/2}(M)\rightarrow H^{m}(N) \text{.}
\end{eqnarray*}
\end{theorem}

Using the extension operator, we generate constraint-satisfying Cauchy data on
$N$ from initial data on $M$ as follows:%
\begin{align*}
u_{0}(x_{1},y_{2}) &  :=E(w_{0})(x_{1},y_{2})+y_{2}E(w_{01})(x_{1},y_{2})\\
u_{1}(x_{1},y_{2}) &  :=E(w_{10})(x_{1},y_{2})\text{.}%
\end{align*}
Checking that this is a legitimate choice, we have
\begin{align*}
u_{1}(x_{1},0) &  =\frac{1}{2\pi}\iint e^{i\xi x_{1}}\hat{w}_{10}(\xi)\chi
(\xi,\eta^{\prime})d\xi d\eta^{\prime}\\
&  =\frac{1}{\sqrt{2\pi}}\int e^{i\xi x_{1}}\hat{w}_{10}(\xi)\left[  \frac
{1}{\sqrt{2\pi}}\int\psi(\eta^{\prime}/\left\vert \xi\right\vert )\frac
{1}{\left\vert \xi\right\vert }d\eta^{\prime}\right]  d\xi\\
&  =\frac{1}{\sqrt{2\pi}}\int e^{i\xi x_{1}}\hat{w}_{10}(\xi)\left[  \frac
{1}{\sqrt{2\pi}}\int\psi(\theta)d\theta\right]  d\xi\\
&  =w_{10}(x_{1})
\end{align*}
because of the normalization (eqn \ref{Normalization}) of $\psi$.\ Similarly,%
\[
u_{0}(x_{1},0)=E(w_{0})(x_{1},0)=w_{0}(x_{1}).
\]
The compatibility condition is satisfied, since%
\begin{align*}
\partial_{y_{2}}u_{0}(x_{1},0) &  =\partial_{y_{2}}E(w_{0})(x_{1}%
,0)+E(w_{0})(x_{1},0)\\
&  =\partial_{y_{2}}E(w_{0})(x_{1},0)+w_{01}(x_{1})\text{.}%
\end{align*}
The first term of the RHS vanishes because
\begin{align*}
\partial_{y_{2}}E(w_{0})(x_{1},0) &  =\frac{1}{2\pi}\iint e^{i\xi x_{1}}%
i\eta^{\prime}\hat{w}_{0}(\xi)\chi(\xi,\eta^{\prime})d\xi d\eta^{\prime}\\
&  =\frac{1}{\sqrt{2\pi}}\int ie^{i\xi x_{1}}\hat{w}_{0}(\xi)\left[  \frac
{1}{\sqrt{2\pi}}\int\eta^{\prime}\psi(\eta^{\prime}/\left\vert \xi\right\vert
)\frac{1}{\left\vert \xi\right\vert }d\eta^{\prime}\right]  d\xi\\
&  =0\text{,}%
\end{align*}
where we have used that $\int\theta\psi(\theta)d\theta=0$ because $\psi
(\cdot)$ has been chosen to be even$.$

The pair of functions $(u_{0}(x_{1},y_{2}),u_{1}(x_{1},y_{2}))$ gives Cauchy
data for the codimension one problem that is discussed in Section 2. Because
of the properties of the extension, it satisfies the constraint conditions of
$X^{C}$ for solutions which are globally defined in $y_{1}$. In order to apply
the existence theorem, the energy norm of this Cauchy data must be finite.

\begin{theorem}
\label{thm5}
Suppose that $w_{0}\in\dot{H}^{1/2}(M)$, $w_{01}\in\dot{H}^{-1/2}$ and
$w_{10}\in\dot{H}^{-1/2}$. Then the energy norm of the extension
$u_{0}=E(w_{0})(x_{1},y_{2})+y_{2}E(w_{01})(x_{1},y_{2})$, $u_{1}(x_{1}%
,y_{2})=E(w_{10})(x_{1},y_{2})$ is finite:%
\[
\left\Vert (u_{0},u_{1})\right\Vert _{X^{C}}^{2}\leq C(\left\Vert
w_{0}\right\Vert _{\dot{H}^{1/2}}^{2}+\left\Vert w_{01}\right\Vert _{\dot
{H}^{-1/2}}^{2}+\left\Vert w_{10}\right\Vert _{\dot{H}^{-1/2}}^{2})\text{.}
\]
Additionally, the higher energy norms with which one defines the $X^{m}$
topology for $(u_{0},u_{1})$ are also bounded by this extension process,
namely
\[
\left\Vert (u_{0},u_{1})\right\Vert _{X^{m,C}}^{2}\leq C_{m}(\left\Vert
w_{0}\right\Vert _{\dot{H}^{m+1/2}}^{2}+\left\Vert w_{01}\right\Vert _{\dot
{H}^{m-1/2}}^{2}+\left\Vert w_{10}\right\Vert _{\dot{H}^{m-1/2}}^{2})\text{.}
\]
\end{theorem}

\begin{proof}
(of Theorem \ref{thm4}): Using the Plancherel identity, the $L^{2}(N)$ norm of $E(w)$
is%
\begin{align*}
\left\Vert E(w)\right\Vert _{L^{2}(N)}^{2}  &  =\iint\left\vert \hat{w}%
(\xi)\right\vert ^{2}\psi^{2}(\eta^{\prime}/\left\vert \xi\right\vert
)\frac{1}{\left\vert \xi\right\vert ^{2}}d\eta^{\prime}d\xi\\
&  =\int\frac{1}{\left\vert \xi\right\vert }\left\vert \hat{w}(\xi)\right\vert
^{2}\left(  \int\psi^{2}(\eta^{\prime}/\left\vert \xi\right\vert )\frac
{1}{\left\vert \xi\right\vert }d\eta^{\prime}\right)  d\xi\\
&  =\left\Vert \psi\right\Vert _{L^{2}[-1,1]}^{2}\left\Vert w\right\Vert
_{\dot{H}^{-1/2}(M)}^{2}\text{,}%
\end{align*}
since $\theta=\eta^{\prime}/\left\vert \xi\right\vert $ and
\[
\int\psi^{2}(\eta^{\prime}/\left\vert \xi\right\vert )\frac{1}{\left\vert
\xi\right\vert }d\eta^{\prime}=\int\limits_{-1}^{1}\psi^{2}(\theta
)d\theta\text{.}%
\]
The identity extends to the Sobolev space $H^{m}(N)$; it suffices to calculate
$\left\Vert \partial_{x_{1}}^{m}E(w)\right\Vert _{L^{2}}$ and $\left\Vert
\partial_{y_{2}}^{m}E(w)\right\Vert _{L^{2}}$:
\begin{align*}
\left\Vert \partial_{x_{1}}^{m}E(w)\right\Vert _{L^{2}(N)}  &  =\iint
\left\vert \hat{w}(\xi)\right\vert ^{2}\left\vert \xi\right\vert ^{2m}\psi
^{2}(\eta^{\prime}/\left\vert \xi\right\vert )\frac{1}{\left\vert
\xi\right\vert ^{2}}d\eta^{\prime}d\xi     \\
&  =\int\left\vert \hat{w}(\xi)\right\vert ^{2}\left\vert \xi\right\vert
^{2m-1}\left(  \int\psi^{2}(\eta^{\prime}/\left\vert \xi\right\vert )\frac
{1}{\left\vert \xi\right\vert }d\eta^{\prime}\right)  d\xi\\
&  =\left\Vert \psi\right\Vert _{L^{2}[-1,1]}^{2}\left\Vert w\right\Vert
_{\dot{H}^{m-1/2}(M)}^{2}\text{.}%
\end{align*}
The second quantity is similar:%
\begin{align*}
\left\Vert \partial_{y_{2}}^{m}E(w)\right\Vert _{L^{2}(N)}  &  =\iint
\left\vert \hat{w}(\xi)\right\vert ^{2}\left\vert \xi\right\vert ^{2m}\psi
^{2}(\eta^{\prime}/\left\vert \xi\right\vert )\frac{1}{\left\vert
\xi\right\vert ^{2}}\left\vert \eta^{\prime}\right\vert ^{2m}d\eta^{\prime
}d\xi\\
&  =\int\left\vert \hat{w}(\xi)\right\vert ^{2}\left\vert \xi\right\vert
^{2m-1}\left(  \int\psi^{2}(\eta^{\prime}/\left\vert \xi\right\vert
)\left\vert \frac{\eta^{\prime}}{\left\vert \xi\right\vert }\right\vert
^{2m}\frac{1}{\left\vert \xi\right\vert }d\eta^{\prime}\right)  d\xi\\
&  =C_{m}\left\Vert w\right\Vert _{\dot{H}^{m-1/2}(M)}^{2}\text{,}%
\end{align*}
where $C_{m}=\int\limits_{-1}^{1}\theta^{2m}\psi^{2}(\theta)d\theta$. The
third and fourth estimates of the theorem involve $y_{2}E(w)$, whose Fourier
transform is%
\[
\hat{w}(\xi)\frac{1}{i}\partial_{\eta^{\prime}}\chi(\eta^{\prime}/\left\vert
\xi\right\vert )\text{.}%
\]
Measuring the $L^{2}$ norm of $y_{2}E(w)$,
\begin{align*}
\left\Vert y_{2}E(w)\right\Vert _{L^{2}(N)}^{2}  &  =\iint\left\vert \hat
{w}(\xi)\right\vert ^{2}\left\vert \frac{1}{i}\partial_{\theta}\psi
(\eta^{\prime}/\left\vert \xi\right\vert )\frac{1}{\left\vert \xi\right\vert
^{2}}\right\vert ^{2}d\eta^{\prime}d\xi\\
&  =\int\left\vert \hat{w}(\xi)\right\vert ^{2}\frac{1}{\left\vert
\xi\right\vert ^{3}}\left(  \int\left\vert \partial_{\theta}\psi(\eta^{\prime
}/\left\vert \xi\right\vert )\right\vert ^{2}\frac{1}{\left\vert
\xi\right\vert }d\eta^{\prime}\right)  d\xi\\
&  =\int\left\vert \hat{w}(\xi)\right\vert ^{2}\frac{1}{\left\vert
\xi\right\vert ^{3}}d\xi\left(  \int\limits_{-1}^{1}\left\vert \partial
_{\theta}\psi\right\vert ^{2}d\theta\right) \\
&  =C\left\Vert w\right\Vert _{\dot{H}^{-3/2}(M)}^{2}%
\end{align*}
with $C=\int\limits_{-1}^{1}\left\vert \partial_{\theta}\psi\right\vert
^{2}d\theta$. The calculations of the $H^{m}$ norms of $y_{2}E(w)$ are similar.
\end{proof}

\begin{proof}
(of Theorem \ref{thm5}): Given initial data $(w_{0},w_{01},w_{10})(x_{1})$ we are to
give conditions under which the energy norm of the extension $(u_{0},u_{1})$
is finite. First of all, the contribution to the energy given by $u_{1}$ is
simply $\frac{1}{2}\left\Vert u_{1}\right\Vert _{L^{2}}^{2}$, hence by Theorem \ref{thm4}
it is bounded by $C\left\Vert w_{10}\right\Vert _{\dot{H}^{-1/2}}^{2}$.
There are two contributions from $u_{0}$, which can be expressed using the
Plancherel identity:%
\[
\iint \omega^{2}(\xi,\eta^{\prime})\left\vert \hat{w}_{0}(\xi)\right\vert ^{2}%
\chi^{2}(\xi,\eta^{\prime})d\eta^{\prime}d\xi+\iint \omega^{2}(\xi,\eta^{\prime
})\left\vert \hat{w}_{01}(\xi)\right\vert ^{2}\left\vert \partial
_{\eta^{\prime}}\chi^{2}(\xi,\eta^{\prime})\right\vert ^{2}d\eta^{\prime}%
d\xi\text{ .}%
\]
Using that $\chi(\xi,\eta^{\prime})=\psi(\eta^{\prime}/\left\vert
\xi\right\vert )\frac{1}{\left\vert \xi\right\vert }$, we estimate these two
integrals:%
\begin{align*}
&  \iint\limits_{\{\left\vert \eta^{\prime}\right\vert <\left\vert
\xi\right\vert \}}\left(  \left\vert \xi\right\vert ^{2}-\left\vert
\eta^{\prime}\right\vert ^{2}\right)  \left\vert \hat{w}_{0}(\xi)\right\vert
^{2}\psi^{2}(\eta^{\prime}/\left\vert \xi\right\vert )\frac{1}{\left\vert
\xi\right\vert ^{2}}d\eta^{\prime}d\xi\text{ }\\
&  =\int\left\vert \hat{w}_{0}(\xi)\right\vert ^{2}\left[  \int\left(
\left\vert \xi\right\vert -\frac{\left\vert \eta^{\prime}\right\vert ^{2}%
}{\left\vert \xi\right\vert }\right)  \psi^{2}(\eta^{\prime}/\left\vert
\xi\right\vert )\frac{1}{\left\vert \xi\right\vert }d\eta^{\prime}\right]
d\xi\text{ }\\
&  \leq C\left\Vert w_{0}\right\Vert _{\dot{H}^{1/2}}^{2}\text{.}%
\end{align*}%
\begin{align*}
&  \iint\limits_{\{\left\vert \eta^{\prime}\right\vert <\left\vert
\xi\right\vert \}}\left(  \left\vert \xi\right\vert ^{2}-\left\vert
\eta^{\prime}\right\vert ^{2}\right)  \left\vert \hat{w}_{01}(\xi)\right\vert
^{2}\psi^{2}(\eta^{\prime}/\left\vert \xi\right\vert )\frac{1}{\left\vert
\xi\right\vert ^{2}}d\eta^{\prime}d\xi\text{ }\\
&  =\int\left\vert \hat{w}_{01}(\xi)\right\vert ^{2}\left[  \int_{\{\left\vert
\eta^{\prime}\right\vert <\left\vert \xi\right\vert \}}\left(  \frac
{1}{\left\vert \xi\right\vert }-\frac{\left\vert \eta^{\prime}\right\vert
^{2}}{\left\vert \xi\right\vert ^{3}}\right)  \psi^{2}(\eta^{\prime
}/\left\vert \xi\right\vert )\frac{1}{\left\vert \xi\right\vert }d\eta
^{\prime}\right]  d\xi\text{ }\\
&  \leq C\left\Vert w_{01}\right\Vert _{\dot{H}^{-1/2}}^{2}\text{.}%
\end{align*}
\end{proof}

\subsection{The extension problem for general spacelike data}

We now consider the general problem of initial data given on a maximal
spacelike hypersurface of dimension $d_{1}$, extending it to Cauchy data on a
codimension one hypersurface. That is, for $(x,y)\in\mathbb{R}_{x}^{d_{1}}
\times\mathbb{R}_{y}^{d_{2}}$,
\[
M=\left\{  y=0\right\}  \subseteq N=\{y_{1}=0\}\text{.}
\]
Initial data on $M$ take the form $w(x)=(w_{0}(x),w_{\alpha}(x))$ where a
solution $u(x,y)$ of the field equation (\ref{ultra}) is asked to satisfy
\[
u(x,y)=w_{0}(x)
\]
with its first derivatives normal to $M$ satisfying%
\[
\partial_{y}^{\alpha}u(x,0)=w_{\alpha}(x)
\]
where $\alpha\in\mathbb{N}^{d_{2}}$ is the multi-index $\alpha=(\alpha
_{1},...,\alpha_{d_{2}}),$ $\left\vert \alpha\right\vert =1$, such that only
one $\alpha_{j}=1$ and the rest are zero. The object is to extend $w(x)$ to
Cauchy data on $N$ while satisfying the constraints (\ref{Constraint}) to be
in $X^{C}$. This Cauchy data satisfies%
\begin{align*}
u_{0}(x,0)  &  =w_{0}(x)\\
u_{\alpha^{\prime}}(x,0)  &  =w_{0\alpha^{\prime}}(x)
\end{align*}
for $\alpha^{\prime}=(\alpha_{2},...,\alpha_{d_{2}})$ and the first
derivatives normal to $N$ satisfy%
\[
\partial_{y_{1}}u(x,0)=w_{10}(x)\text{ .}%
\]

Following the construction given in section 3.1, define an extension operator%
\[
E(w)(x,y^{\prime}):=\frac{1}{\sqrt{2\pi}^{d_{1}+d_{2}-1}}\iint\hat{w}(\xi
)\chi(\xi,\eta^{\prime})e^{i\xi\cdot x}e^{i\eta^{\prime}\cdot y^{\prime}}d\xi
d\eta^{\prime}%
\]
where the kernel function is even in $\eta$ and satisfies%
\[
\frac{1}{\sqrt{2\pi}^{d_{2}-1}}\int\chi(\xi,\eta^{\prime})d\eta^{\prime
}=1\text{ .}%
\]
To satisfy the constraint that $E(w)\in X^{C}$, we ask that $\mathrm{supp(}%
\chi(\xi,\eta^{\prime}))\subseteq\{(\xi,\eta^{\prime}):\left\vert \eta
^{\prime}\right\vert <\left\vert \xi\right\vert \}$. Such kernel functions are
readily constructed (they are far from being uniquely determined). For
example, a variant of our construction of section 3.1 is based on choice of a
$C_{0}^{\infty}$ function $\psi(\theta)\geq0$, with $\mathrm{supp}%
(\psi)\subseteq B_{1}(0)$, the ball of radius one. Then define
\[
\chi(\xi,\eta^{\prime})=\psi(\eta^{\prime}/\left\vert \xi\right\vert )\frac
{1}{\left\vert \xi\right\vert ^{d_{2}-1}}~\text{.}%
\]
We note that $\chi$ is even in $\eta$ if $\psi(\theta)$ is even, and that%
\begin{align*}
\frac{1}{\sqrt{2\pi}^{d_{2}-1}}\int\chi(\xi,\eta^{\prime})d\eta^{\prime}  &
=\frac{1}{\sqrt{2\pi}^{d_{2}-1}}\int\psi(\eta^{\prime}/\left\vert
\xi\right\vert )\frac{1}{\left\vert \xi\right\vert ^{d_{2}-1}}d\eta^{\prime}\\
&  =\frac{1}{\sqrt{2\pi}^{d_{2}-1}}\int\psi(\theta)d\theta\text{ .}%
\end{align*}
This is normalized to one by choice of $\psi$.

\begin{theorem}
\label{thm6}
The extension operator $E$ is bounded on the following function spaces:
\begin{eqnarray*}
  && E:\dot{H}^{\frac{1-d_{2}}{2}}(M)\rightarrow L^{2}(N)  \\
  && \quad \ \dot{H}^{\frac{1-d_{2}}{2}}(M)\cap
  H^{m+\frac{1-d_{2}}{2}}(M)  \rightarrow H^{m}(N) ~,
\end{eqnarray*}
with $m$ the exponent of Sobolev regularity, and
\begin{eqnarray*}
  && y^{\prime}E:\dot{H}^{\frac{-(1+d_{2})}{2}}(M)\rightarrow L^{2}(N)\\
\dot{H}^{\frac{-(1+d_{2})}{2}}(M)\cap H^{m-\frac{1+d_{2}}{2}}(M)\rightarrow
H^{m}(N)\text{ .}
\end{eqnarray*}
\end{theorem}

Using the extension operator $E$, the vector function $w(x)=(w_{0}%
(x),w_{\alpha}(x))$ extends to Cauchy data on $N$ as follows:%
\begin{align*}
u_{0}(x,y^{\prime}) &  :=E(w_{0})(x,y^{\prime})+\sum\limits_{\left\vert
\alpha^{\prime}\right\vert =1}y^{\alpha^{\prime}}\cdot E(w_{0\alpha^{\prime}%
})(x,y^{\prime})\\
u_{1}(x,y^{\prime}) &  :=E(w_{10})(x,y^{\prime})\text{ .}%
\end{align*}
This is seen to extend the initial data $w(x)$ in the required way, and in
addition it satisfies the constraint that $(u_{0},u_{1})\in X^{C}$. However,
measuring the functions $(u_{0},u_{1})$ in the energy norm is more appropriate
for the Cauchy problem, hence we also state estimates in this setting.

\begin{theorem}
\label{thm7}
Given $w_{0}\in\dot{H}(M)$ and $w_{\alpha}\in\dot{H}(M)$, the energy norm of
the extension
\[
(u_{0},u_{1})=(E(w_{0})+y^{\alpha^{\prime}}\cdot E(w_{(0,\alpha^{\prime}%
)}),E(w_{10}))
\]
is finite; indeed
\[
\left\Vert (u_{0},u_{1})\right\Vert _{X^{C}}^{2}\leq C(\left\Vert
w_{0}\right\Vert _{\dot{H}^{\frac{3-d_{2}}{2}}}^{2}+\left\Vert w_{0\alpha
^{\prime}}\right\Vert _{\dot{H}^{\frac{1-d_{2}}{2}}}^{2}+\left\Vert
w_{10}\right\Vert _{\dot{H}^{\frac{1-d_{2}}{2}}}^{2})  ~.
\]
\end{theorem}

The proofs of Theorems \ref{thm6} and \ref{thm7} are similar to the proofs of Theorems \ref{thm4} and \ref{thm5}, to which we refer the reader.

\subsection{The extension problem for mixed spacelike and timelike data}

As a final case, we consider the extension problem for initial data on a lower dimensional hypersurface $M$ of \emph{mixed} signature. Given zero'th and first normal derivatives of a solution $u(x,y)$
on $M$, the object is to extend this data to the codimension one hypersurface
$N=\{y_{1}=0\}$ in such a way that the constraint for well-posedness is
satisfied. This is not always possible for arbitrary data $w=(w_{0},w_{\alpha
})$ posed on $M$, due to analogous lower dimensional constraints on $M$. But
it is possible, along with attendant Sobolev bounds on the extended functions,
in most cases. This situation will be explained below.

To set the notation, we consider spacelike and timelike coordinates on $M$ to
be $(\tilde{x},\tilde{y})\in\mathbb{R}^{p_{1}}\times\mathbb{R}^{p_{2}}$, with
their Fourier transform variables denoted $(\tilde{\xi},\tilde{\eta}%
)\in\mathbb{R}^{p_{1}}\times\mathbb{R}^{p_{2}}$. The complementary variables
will be denoted $(x^{\prime\prime},y^{\prime\prime})\in\mathbb{R}^{d_{1}%
-p_{1}}\times\mathbb{R}^{d_{2}-p_{2}-1}$ and $(\xi^{\prime\prime},\eta
^{\prime\prime})\in\mathbb{R}^{d_{1}-p_{1}}\times\mathbb{R}^{d_{2}-p_{2}-1}$,
so that coordinates on $N$ are $(x,y^{\prime})=(\tilde{x},x^{\prime\prime
},\tilde{y},y^{\prime\prime})$. The evolution variable remains $y_{1}$.

Initial data for a solution $u(x,y)$ is given on $N$, which is expressed in
the form $(u,\partial_{x^{\prime\prime}}^{\alpha^{\prime\prime}}%
u,\partial_{y_{1}}^{\beta_{1}}u,\partial_{y^{\prime\prime}}^{\beta
^{\prime\prime}}u)(\tilde{x},\tilde{y},0,0)=(w_{0},w_{\alpha^{\prime\prime}%
},w_{\beta_{1}},w_{\beta^{\prime\prime}})(\tilde{x},\tilde{y})$, where
$\alpha^{\prime\prime}=(\alpha_{p_{1}+1},...,\alpha_{d_{1}})$, $\beta
^{\prime\prime}=(\beta_{p_{2}+1},...,\beta_{d_{2}})$ are multi-indices such
that $\left\vert \alpha^{\prime\prime}\right\vert +\left\vert \beta
^{\prime\prime}\right\vert +\left\vert \beta_{1}\right\vert =1$. The idea is
the same as in sections 3.1 and 3.2, namely to extend $(w_{0},w_{\alpha
^{\prime\prime}},w_{\beta_{1}},w_{\beta^{\prime\prime}})$ to
constraint-satisfying Cauchy data on $N$ in such a way that a solution
$u(x,y)=u(\tilde{x},x^{\prime\prime},\tilde{y},y^{\prime\prime})$ to the field
equation (\ref{ultra}) satisfies%
\[
u(\tilde{x},0,0,\tilde{y},0)=w_{0}(\tilde{x},\tilde{y})
\]
and%
\[
\partial_{y_{1}}u(\tilde{x},0,0,\tilde{y},0)=w_{0\beta_{1}}(\tilde{x}%
,\tilde{y})~,
\]
as well as the compatibility conditions%
\begin{align*}
\partial_{x^{\prime\prime}}^{\alpha^{\prime\prime}}u(\tilde{x},0,0,\tilde
{y},0) &  =w_{(\alpha^{\prime\prime},0)}(\tilde{x},\tilde{y})\\
\partial_{y^{\prime\prime}}^{\beta^{\prime\prime}}u(\tilde{x},0,0,\tilde{y},0)
&  =w_{(0,\beta^{\prime\prime})}(\tilde{x},\tilde{y})
\end{align*}
The existence of such an extension follows as in Theorems \ref{thm6} and \ref{thm7} from the
construction of an extension operator $E$ with certain boundedness properties
on appropriate Sobolev spaces. We will focus our analysis therefore on the
extension operators.

Again following section 3.1, define an extension operator
\[
E(w)(x,y^{\prime})=\frac{1}{\sqrt{2\pi}^{d_{1}+d_{2}-1}}\iint\chi(\tilde{\xi
},\xi^{\prime\prime},\tilde{\eta},\eta^{\prime\prime})d\xi^{\prime\prime}%
d\eta^{\prime\prime}=1\text{.}%
\]
Furthermore, to satisfy the constraint that $E(w)\in X^{C}$ for arbitrary data
$w$, we ask that
\[
\mathrm{supp(\chi(\xi,\eta}^{\prime}))\subseteq\left\{  (\xi,\eta^{\prime
}):\left\vert \eta^{\prime}\right\vert ^{2}<\left\vert \xi\right\vert
^{2}\right\}  :=R_{1}\text{.}%
\]
These two conditions are always satisfiable, except in the case 
$\xi^{\prime\prime}=\left\{  0\right\}  $, meaning that $d_{1}=p_{1}$ and the
extension subspace $\left\{  (\xi^{\prime\prime},\eta^{\prime\prime})\right\}
$ is purely timelike.

It is to be expected that the constraint induces a restriction on the data
$w(\tilde{\xi},\tilde{\eta})$ in the vicinity of the ``lightcone'' $\left\{
\left\vert \tilde{\xi}\right\vert =\left\vert \tilde{\eta}\right\vert
\right\}  \subseteq\hat{M}$. Subdivide $\hat{M}$ into two sets,%
\begin{align*}
\tilde{R}_{1}  &  :=\left\{  (\tilde{\xi},\tilde{\eta})\in\hat{M}:\left\vert
\tilde{\eta}\right\vert \leq\left\vert \tilde{\xi}\right\vert \right\} \\
\tilde{R}_{2}  &  :=\left\{  (\tilde{\xi},\tilde{\eta})\in\hat{M}:\left\vert
\tilde{\eta}\right\vert >\left\vert \tilde{\xi}\right\vert \right\}  \text{.}%
\end{align*}
The orthogonal projections onto functions supported in $\tilde{R}_{1}$,
$\tilde{R}_{2}$ respectively, are denoted $\pi_{1}$ and $\pi_{2}$. We use
standard Sobolev spaces to quantify data supported in $\tilde{R}_{1}$, namely
\[
H^{r}=\left\{  w(\tilde{x},\tilde{y})\in{\mathrm{range}}(\pi_{1}):\left\Vert
w\right\Vert _{H^{r}}^{2}=\iint\limits_{\tilde{R}_{1}}\left\vert \hat
{w}(\tilde{\xi},\tilde{\eta})\right\vert ^{2}(\left\vert \tilde{\xi
}\right\vert ^{2}+\left\vert \tilde{\eta}\right\vert ^{2})^{n}d\tilde{\xi
}d\tilde{\eta}<+\infty\right\}  ~\text{.}%
\]
Over $\tilde{R}_{2}$ we use a modified form of Sobolev norm which is given by
\[
K^{r}=\left\{  w(\tilde{x},\tilde{y})\in{\mathrm{range}}(\pi_{2}):\left\Vert
w\right\Vert _{K^{r}}^{2}=\iint\limits_{\tilde{R}_{2}}\left\vert \hat
{w}(\tilde{\xi},\tilde{\eta})\right\vert ^{2}\frac{(\left\vert \tilde{\xi
}\right\vert ^{2}+\left\vert \tilde{\eta}\right\vert ^{2})^{r}}{(\left\vert
\tilde{\eta}\right\vert ^{2}-\left\vert \tilde{\xi}\right\vert ^{2})^{\frac
{1}{2}e_{0}}}d\tilde{\xi}d\tilde{\eta}<+\infty\right\}  \text{.}%
\]
where
\[
e_{0}:=d_{1}+d_{2}-(p_{1}+p_{2})-1.
\]
We note that in the case where $d_{1}=p_{1}$, $\tilde{R}_{1}=\left\{
0\right\}  $ and $K^{n}=H^{r-\frac{r}{2}(d_{2}-p_{2}-1)}$. More generally,
define%
\[
K_{s}^{r}=\left\{  w(\tilde{x},\tilde{y})\in{\mathrm{range}}(\pi
_{2}):\left\Vert w\right\Vert _{K_{s}^{r}}^{2}=\iint\limits_{\tilde{R}_{2}%
}\left\vert \hat{w}(\tilde{\xi},\tilde{\eta})\right\vert ^{2}\frac{(\left\vert
\tilde{\xi}\right\vert ^{2}+\left\vert \tilde{\eta}\right\vert ^{2})^{r}%
}{(\left\vert \tilde{\eta}\right\vert ^{2}-\left\vert \tilde{\xi}\right\vert
^{2})^{\frac{1}{2}e_{0}+s}}d\tilde{\xi}d\tilde{\eta}<+\infty\right\}
\]
Decompose an arbitrary function $w=\pi_{1}w+\pi_{2}w$, so that its components possess Fourier support in $\tilde{R}_{1}$ and $\tilde{R}_{2}$ respectively.

\begin{theorem}
\label{thm8}
If $d_{1}>p_{1}$ then there is a choice of kernel $\chi$ (indeed there are
many such choices) such that $u=E(w)$ satisfies%
\[
\left\Vert u\right\Vert _{L^{2}}^{2}\leq C(\left\Vert \pi_{1}w\right\Vert
_{H^{-\frac{1}{2}(e_{0})}}^{2}+\left\Vert \pi_{2}w\right\Vert _{K^{0}}%
^{2})\text{.}%
\]
Higher Sobolev norms of $u=E(w)$ are bounded as follows%
\[
\left\Vert u\right\Vert _{H^{r}}^{2}\leq C_{r}(\left\Vert \pi_{1}w\right\Vert
_{H^{r-\frac{1}{2}(e_{0})}}^{2}+\left\Vert \pi_{2}w\right\Vert _{K^{r}}%
^{2}\text{.}%
\]
In case $d_{1}=p_{1}$, it is not possible to extend arbitrary data to a
function $u=E(w)$ which satisfies the constraint $\mathrm{supp}(\hat{u}%
(\xi,\eta^{\prime}))\subseteq R_{1}$. However, if initially $\mathrm{supp}%
(\hat{w}(\xi,\eta^{\prime}))\subseteq\tilde{R}_{1}$ (i.e., $w=\pi_{2}w$), then
such an extension is possible, and we have, for $u=E(w)$,%
\[
\left\Vert u\right\Vert _{L^{2}}^{2}\leq C\left\Vert w\right\Vert _{K^{0}}%
^{2}~\text{,}%
\]%
\[
\left\Vert u\right\Vert _{H^{r}}^{2}\leq C_{r}\left\Vert w\right\Vert _{K^{r}%
}^{2}~\text{.}%
\]
\end{theorem}

\begin{proof}
The proof of Theorem \ref{thm8} depends upon the construction of a kernel $\chi
(\tilde{\xi},\tilde{\eta},\xi^{\prime\prime},\eta^{\prime\prime})$ with
satisfactory properties. This construction is slightly different in the two
different regions of Fourier space%
\[
\tilde{R}_{1}:=\left\{  (\tilde{\xi},\tilde{\eta}):\left\vert \tilde{\eta
}\right\vert \leq\left\vert \tilde{\xi}\right\vert \right\}  \text{ and
}\tilde{R}_{2}:=\left\{  (\tilde{\xi},\tilde{\eta}):\left\vert \tilde{\eta
}\right\vert >\left\vert \tilde{\xi}\right\vert \right\}
\]
where we note that the region $\tilde{R}_{2}$ contains the data which would
lead to an ill-posed initial value problem if $M$ were considered itself as a
codimension one hypersurface.

To extend data posed on region $\tilde{R}_{1}$, define
\[
\chi_{1}(\tilde{\xi},\tilde{\eta},\xi^{\prime\prime},\eta^{\prime\prime
}):= \psi_{1}\Bigl(\frac{\xi^{\prime\prime}}{(\left\vert \tilde{\xi}\right\vert
^{2}+\left\vert \tilde{\eta}\right\vert ^{2})^{\frac{1}{2}}},\frac
{\eta^{\prime\prime}}{(\left\vert \tilde{\xi}\right\vert ^{2}+\left\vert
\tilde{\eta}\right\vert ^{2})^{\frac{1}{2}}}\Bigr)\cdot\frac{1}{(\left\vert
\tilde{\xi}\right\vert ^{2}+\left\vert \tilde{\eta}\right\vert ^{2})^{\frac
{1}{2}e_{0}}}\text{,}%
\]
where $\psi_{1}(\theta_{1},\theta_{2})$ is a $C_{0}^{\infty}$ function of
$(d_{1}-p_{1})\times(d_{2}-p_{2}-1)$ variables, respectively, with support in
the set $\left\vert \theta_{2}\right\vert <\left\vert \theta_{1}\right\vert $.
Therefore $\chi_{1}$ has support in the set
\[
\frac{\xi^{\prime\prime}}{(\left\vert \tilde{\xi}\right\vert ^{2}+\left\vert
\tilde{\eta}\right\vert ^{2})^{\frac{1}{2}}}\geq\frac{\eta^{\prime\prime}%
}{(\left\vert \tilde{\xi}\right\vert ^{2}+\left\vert \tilde{\eta}\right\vert
^{2})^{\frac{1}{2}}}
\]
implying that
\[
\left\vert \eta^{\prime}\right\vert ^{2}=\left\vert \tilde{\eta}\right\vert
^{2}+\left\vert \eta^{^{\prime\prime}}\right\vert ^{2}<\left\vert \tilde{\xi
}\right\vert ^{2}+\left\vert \xi^{\prime\prime}\right\vert ^{2}=\left\vert
\xi\right\vert ^{2}\text{.}
\]
This is the appropriate region of support from functions $v=E(w)$ to lie in
the constraint-satisfying subspace of $L^{2}(N)$. In order that $E$ be an
extension operator, we furthermore require that
\begin{align*}
\sqrt{2\pi}^{d_{1}+d_{2}-1} &  =\iint\chi_{1}(\tilde{\xi},\tilde{\eta}%
,\xi^{\prime\prime},\eta^{\prime\prime})d\xi^{\prime\prime}d\eta^{\prime
\prime}\\
&  =\iint\psi_{1}\Bigl(\frac{\xi^{\prime\prime}}{(\left\vert \tilde{\xi}\right\vert
^{2}+\left\vert \tilde{\eta}\right\vert ^{2})^{\frac{1}{2}}},\frac
{\eta^{\prime\prime}}{(\left\vert \tilde{\xi}\right\vert ^{2}+\left\vert
\tilde{\eta}\right\vert ^{2})^{\frac{1}{2}}}\Bigr)\cdot\frac{1}{(\left\vert
\tilde{\xi}\right\vert ^{2}+\left\vert \tilde{\eta}\right\vert ^{2})^{\frac
{1}{2}e_{0}}}d\xi^{\prime\prime}d\eta^{\prime\prime}\\
&  =\iint\psi_{1}(\theta_{1},\theta_{2})d\theta_{1}d\theta_{2}\text{.}%
\end{align*}
Asking that this latter integral equal the normalizing constant $\sqrt{2\pi
}^{d_{1}+d_{2}-1}$, and asking for $\psi_{1}$ to be even in its variables
$(\theta_{1},\theta_{2})$ gives an acceptable kernel for the extension
operator. We note again that this choice of kernel is highly nonunique.

On the region $\tilde{R}_{2}=\left\{  (\tilde{\xi},\tilde{\eta})\in\hat
{M}:\left\vert \tilde{\eta}\right\vert >\left\vert \tilde{\xi}\right\vert
\right\}  $, we can also attempt a construction of our extension operator.
By itself, this region would give rise to data in $L^{2}(M)$ for which the Cauchy problem
of mixed type is ill-posed. The extension operator will nonetheless come up
with data $u=E(w)$ for which the well-posedness constraint is satisfied, if
this is possible. That is, as long as $d_{1}>p_{1}$, so that $\left\{
\xi^{\prime\prime}\right\}  $ is not restricted to the zero-dimensional vector
space, extensions can be found in a way that the default in satisfying the
constraint caused by the fact that $\left\vert \tilde{\xi}\right\vert
<\left\vert \tilde{\eta}\right\vert $ can be made up with a choice of large
$\left\vert \xi^{\prime\prime}\right\vert $.
In practice, we will build $\chi_{2}(\tilde{\xi},\tilde{\eta},\xi
^{\prime\prime}, \eta^{\prime\prime})$ so that its support is in the regions
\[
\left\{  \left\vert \tilde{\eta}\right\vert >\left\vert \tilde{\xi}\right\vert
\right\}  = \tilde{R}_{2}%
\]
as well as
\[
\left\{  \left\vert \tilde{\eta}\right\vert ^{2}+\left\vert \eta^{\prime
\prime}\right\vert ^{2}<\left\vert \tilde{\xi}\right\vert ^{2}+\left\vert
\xi^{\prime\prime}\right\vert ^{2}\right\}  ;
\]
implying that $0\leq(\left\vert \tilde{\eta}\right\vert ^{2}-\left\vert
\tilde{\xi}\right\vert ^{2})+(\left\vert \eta^{\prime\prime}\right\vert
^{2}+\left\vert \xi^{\prime\prime}\right\vert ^{2})$. Thus we require
$d_{1}>p_{1}$.
Following the above examples, assume that $d_{1}>p_{1}$ and set%
\[
\chi_{2}(\tilde{\xi},\tilde{\eta},\xi^{\prime\prime},\eta^{\prime\prime
}):=\psi_{2}\Bigl(\frac{\xi^{\prime\prime}}{(\left\vert \tilde{\eta}\right\vert
^{2}-\left\vert \tilde{\xi}\right\vert ^{2})^{\frac{1}{2}}},\frac{\eta
^{\prime\prime}}{(\left\vert \tilde{\eta}\right\vert ^{2}-\left\vert
\tilde{\xi}\right\vert ^{2})^{\frac{1}{2}}}\Bigr)\cdot\frac{1}{(\left\vert
\tilde{\eta}\right\vert ^{2}-\left\vert \tilde{\xi}\right\vert ^{2})^{\frac
{1}{2}e_{0}}}%
\]
for $(\tilde{\xi},\tilde{\eta})\in\tilde{R}_{2}$. Let $\psi_{2}(\theta
_{1},\theta_{2})$ be a $C_{0}^{\infty}$ function of $e_{0}=d_{1}+d_{2}%
-(p_{1}+p_{2})-1$ variables, as before and require that%
\begin{align*}
\int\psi_{2}(\theta_{1},\theta_{2})d\theta_{1}d\theta_{2} & =\int\psi_{2}
\Bigl(\frac{\xi^{\prime\prime}}{(\left\vert \tilde{\eta}\right\vert
^{2}-\left\vert \tilde{\xi}\right\vert ^{2})^{\frac{1}{2}}},\frac{\eta
^{\prime\prime}}{(\left\vert \tilde{\eta}\right\vert ^{2}-\left\vert
\tilde{\xi}\right\vert ^{2})^{\frac{1}{2}}}\Bigr)\cdot\frac{1}{(\left\vert
\tilde{\eta}\right\vert ^{2}-\left\vert \tilde{\xi}\right\vert ^{2})^{\frac
{1}{2}e_{0}}}d\xi^{\prime\prime}d\eta^{\prime\prime}\\
&  =\sqrt{2\pi}^{e_{0}}\text{.}%
\end{align*}
Furthermore, ask that $\psi(\theta_{1},\theta_{2})$ be even in $(\theta
_{1},\theta_{2})$. Finally ask that the support of $\psi(\theta_{1},\theta
_{2})$ be in the set%
\[
\left\{  (\theta_{1},\theta_{2}):\theta_{1}^{2}-\theta_{2}^{2}>1\right\}
\text{.}%
\]
Such requirements are satisfied by many possible choices of $\psi$. In doing
so, we arrive at a satisfactory kernel of an extension operator $E$ with the
property that all functions $u=E(w)$ in its range have Fourier support
satisfying $\mathrm{supp}(\hat{u})\leq\left\{  \left\vert \eta^{\prime
}\right\vert ^{2}<\left\vert \xi\right\vert ^{2}\right\}  $.
The singularities introduced at the boundaries of the lightcone $\left\{
\left\vert \tilde{\eta}\right\vert =\left\vert \tilde{\xi}\right\vert
\right\}  \subseteq\hat{M}$ by the kernel $\chi_{2}$ impose more severe
constraints on the functions $w$ that are permitted in the domain of the
operator $E$; this is the origin of the somewhat unusual requirements on
functions $w(\tilde{x},\tilde{y})$ from which we can reasonably draw our data.
The Sobolev estimates of the proof are similar to those of Theorems \ref{thm6} and \ref{thm7} and we leave the details to the reader.
\end{proof}

Finally, we show that a sufficiently large class of data $(w_{0}%
,w_{(\alpha^{\prime\prime},0)},w_{(0,\beta^{\prime\prime})},w_{\beta_{1}})$ on
$M$ extends to Cauchy data on the hypersurface $N$ which is both of finite
energy and satisfies the constraint. This extension is given by
\begin{align}
u_{0}(x,y^{\prime}) &  =E(w_{0})(x,y^{\prime})+\sum\limits_{\left\vert
\alpha^{\prime\prime}\right\vert =1}x^{\prime\prime\alpha^{\prime\prime}%
}E(w_{(\alpha^{\prime\prime},0)})(x,y^{\prime})+\sum\limits_{\left\vert
\beta^{\prime\prime}\right\vert =1}y^{\prime\prime\beta^{\prime\prime}%
}E(w_{(0,\beta^{\prime\prime})})(x,y^{\prime})\label{Extension}\\
u_{1}(x,y^{\prime}) &  =E(w_{(0,\beta_{1})})(x,y^{\prime})\text{.}\nonumber
\end{align}
By design, this Cauchy data satisfies the constraint, that is, $(u_{0}%
,u_{1})\in X^{C}$, the center manifold. As before, its restriction to $M$
reduces to the data $(w_{0},w_{(\alpha^{\prime\prime},0)},w_{(0,\beta
^{\prime\prime})})(\tilde{x},\tilde{y},0)$. The only remaining task is to show
that its energy norm is finite. Recall that in this context the energy norm
is
\begin{align*}
H(u_{0},u_{1}) &  =\frac{1}{2}\iint\limits_{N}\left\vert u_{1}\right\vert
^{2}+\left\vert \nabla_{x}u_{0}\right\vert ^{2}-\left\vert \nabla_{y^{\prime}%
}u_{0}\right\vert ^{2}dxdy^{\prime}\\
&  =\frac{1}{2}\iint\limits_{N}\left\vert \hat{u}_{1}(\xi,\eta^{\prime
})\right\vert ^{2}+(\left\vert \xi\right\vert ^{2}-\left\vert \eta^{\prime
}\right\vert ^{2})\hat{u}_{0}(\xi,\eta^{\prime})d\xi d\eta^{\prime}\text{.}%
\end{align*}
To show that this energy is finite for the extension (\ref{Extension}), we use
the results of Theorem \ref{thm8}.

\begin{theorem}
\label{thm9}
Given data $(w_{0},w_{(\alpha^{\prime\prime},0)},w_{(0,\beta^{\prime\prime}%
)},w_{\beta_{1}})$ on $M$ with $\left\vert \alpha^{\prime\prime}\right\vert
=\left\vert \beta^{\prime\prime}\right\vert =1$, suppose that
\begin{align}
&  \left\Vert \pi_{1}w_{0}\right\Vert _{H^{e_{0}+1}}+\sum\limits_{\left\vert
\alpha^{\prime\prime}\right\vert =1}\left\Vert \pi_{1}w_{(\alpha^{\prime
\prime},0)}\right\Vert _{H^{e_{0}+1}}+\sum\limits_{\left\vert \beta
^{\prime\prime}\right\vert =1}\left\Vert \pi_{1}w_{(0,\beta^{\prime\prime}%
)}\right\Vert _{H^{e_{0}+1}}<+\infty\\
&  \left\Vert \pi_{2}w_{0}\right\Vert _{K^{1}}+\sum\limits_{\left\vert
\alpha^{\prime\prime}\right\vert =1}\left\Vert \pi_{2}w_{(\alpha^{\prime
\prime},0)}\right\Vert _{K^{1}}+\sum\limits_{\left\vert \beta^{\prime\prime
}\right\vert =1}\left\Vert \pi_{2}w_{(0,\beta^{\prime\prime})}\right\Vert
_{K^{1}}<+\infty
\end{align}
and%
\[
\left\Vert \pi_{1}w_{\beta_{1}}\right\Vert _{H^{e_{0}}}+\left\Vert \pi
_{2}w_{\beta_{1}}\right\Vert _{K^{0}}+\infty\text{.}%
\]
Then the extension $(u_{0},u_{1})$ given by expression (\ref{Extension}) has
finite energy and lies in the center subspace $X^{C}$. If $d_{1}=p_{1}$, then
we have to ask that $\pi_{2}w_{\gamma}=0$ in the above statement, for all
multi-indices $\gamma$ in question.
\end{theorem}

\begin{proof}
Estimates on the contributions of $w_{0}$ to $u_{0}$ follow immediately from
Theorem \ref{thm8}, as do the estimates for $u_{1}=E(w_{\beta_{1}})$. Therefore we
only have to consider contributions in one of the two possible forms:%
\[%
\begin{array}
[c]{cc}%
x^{\prime\prime\alpha^{\prime\prime}}E(w_{(\alpha^{\prime\prime},0)}) &
\qquad\left\vert \alpha^{\prime\prime}\right\vert =1
\end{array}
\]
or%
\[%
\begin{array}
[c]{cc}%
y^{\prime\prime\beta^{\prime\prime}}E(w_{(0,\beta^{\prime\prime})}) &
\qquad\left\vert \beta^{\prime\prime}\right\vert =1
\end{array}
\text{.}%
\]
The energy norm includes the quantities $\left\Vert x^{\prime\prime
\alpha^{\prime\prime}}E(w_{(\alpha^{\prime\prime},0)})\right\Vert _{H^{1}}$
and $\left\Vert y^{\prime\prime\beta^{\prime\prime}}E(w_{(0,\beta
^{\prime\prime})})\right\Vert _{H^{1}}$; since the estimates are similar we
will give a sketch of one of them.
\begin{align*}
& \left\Vert x^{\prime\prime\alpha^{\prime\prime}}E(w_{(\alpha^{\prime\prime
},0)})\right\Vert _{H^{1}}^{2} =\left\Vert \frac{1}{i}\partial_{\xi
^{\prime\prime}}^{a^{\prime\prime}}\widehat{E(w_{(\alpha^{\prime\prime},0)}
)}(\left\vert \xi\right\vert ^{2}+\left\vert \eta^{\prime}\right\vert
^{2})^{\frac{1}{2}}\right\Vert _{L^{2}}^{2}  \\
&
\leq\iiiint\Bigl[\partial_{\xi^{\prime\prime}}^{a^{\prime\prime}}\psi_{1}
\Bigl(\frac{\xi^{\prime\prime}}{(\left\vert \tilde{\xi}\right\vert
^{2}+\left\vert \tilde{\eta}\right\vert ^{2})^{\frac{1}{2}}},\frac
{\eta^{\prime\prime}}{(\left\vert \tilde{\xi}\right\vert ^{2}+\left\vert
\tilde{\eta}\right\vert ^{2})^{\frac{1}{2}}}\Bigr)\cdot\frac{1}{(\left\vert
\tilde{\xi}\right\vert ^{2}+\left\vert \tilde{\eta}\right\vert ^{2})^{\frac
{1}{2}e_{0}}}\Bigr]^{2}\left\vert \widehat{\pi_{1}w_{(\alpha^{\prime\prime
},0)}}(\tilde{\xi},\tilde{\eta})\right\vert ^{2}  \\
&  +\Bigl[\partial_{\xi^{\prime\prime}}^{a^{\prime\prime}}\psi_{2}\Bigl(\frac
{\xi^{\prime\prime}}{(\left\vert \tilde{\eta}\right\vert ^{2}-\left\vert
\tilde{\xi}\right\vert ^{2})^{\frac{1}{2}}},\frac{\eta^{\prime\prime}
}{(\left\vert \tilde{\eta}\right\vert ^{2}-\left\vert \tilde{\xi}\right\vert
^{2})^{\frac{1}{2}}}\Bigr)\cdot\frac{1}{(\left\vert \tilde{\eta}\right\vert
^{2}-\left\vert \tilde{\xi}\right\vert ^{2})^{\frac{1}{2}e_{0}}}
\Bigr]^{2}\left\vert \widehat{\pi_{2}w_{(\alpha^{\prime\prime},0)}}(\tilde
{\xi},\tilde{\eta})\right\vert ^{2}d\tilde{\xi}d\tilde{\eta}d\xi^{\prime
\prime}d\eta^{\prime\prime}\text{.}
\end{align*}
The $\xi^{\prime\prime}$-derivative introduces one extra factor of
$(\left\vert \tilde{\xi}\right\vert ^{2}+\left\vert \tilde{\eta}\right\vert
^{2})$, respectively $(\left\vert \tilde{\eta}\right\vert ^{2}-\left\vert
\tilde{\xi}\right\vert ^{2})$, into the denominator. The integral over
$(\xi^{\prime\prime},\eta^{\prime\prime})$ gives a constant, depending upon
$\psi_{1}$ and $\psi_{2}$, as a bound, while the resulting integral over the
variable $(\tilde{\xi},\tilde{\eta})$ is bounded by the $H^{1-e_{0}}$ norm
(respectively, the $K_{1}^{1}~$norm) of $w_{(\alpha^{\prime\prime},0)}$. This
finishes the proof.
\end{proof}

\section{Failure of uniqueness in higher codimension}

\label{Sec:FailureUniqueness}

The question addressed in this section is the uniqueness of solutions with
prescribed initial data on a hypersurface $M$ of codimension greater than one.
This is a nontrivial issue if one requires that solutions exist globally in
space-time, which has been the focus of the analysis in the preceding sections.
In section 3 we showed that initial data consisting of the values of the
solution $u(x,y)$ and its first normal derivatives on $M$, through a procedure
of extension, give rise to constraint-satisfying Cauchy data on a
codimension-one hypersurface $N$. These extensions are highly nonunique, and
therefore so are the resulting global solutions.

We now raise the question whether prescribing an arbitrarily large but finite
number of normal derivatives on $M$, as well as insisting upon global
solutions, would remedy the nonuniqueness. This data should satisfy the
compatibility conditions implied by the commuting of mixed partial derivatives
and by equation \eqref{ultra}. Given Courant's classic result (1962) in the
case of purely timelike $M$, that data given in any $\varepsilon$-tubular
neighborhood of $M$ within $N$ determine solutions uniquely in the $C^{2}$
category, one might think that specifying additional data for $u(x,y)$ on $M$
would suffice. In fact, if one specifies any finite number of derivatives of
$u$ on $M$ it does not.

\begin{theorem}
\label{thm10}
Given $k$, there exist constraint-satisfying data $u_{0},u_{1}$ on $N$ which
vanish to order $k$ on $M$.
\end{theorem}

\noindent Therefore, there exists a globally defined solution $u(x,y)$ which
has initial data $u(x,y)=u_{0}$, $\partial_{y_{1}}u(x,y)=u_{1}$ on $N$, which
vanishes to order $k$ on $M$. Hence any other solution $v(x,y)$ which takes on
specified data on $M$ up to $k$-many derivatives may be changed by adding this
solution $u$ to it, without changing its initial data.

\begin{proof}
We follow a construction that was used for the extension operators of section
$3$. Let $\chi_{3}(\xi,\eta^{\prime})$ be a Schwartz class function with
support in the set $\left\{  \left\vert \eta^{\prime}\right\vert
^{2}<\left\vert \xi\right\vert ^{2}\right\}  \subseteq\hat{N}$. Its Fourier
restriction to $\hat{M}$, given by
\[
\iint\chi_{3}(\tilde{\xi},\tilde{\eta},\xi^{\prime\prime},\eta^{\prime\prime
})d\xi^{\prime\prime}d\eta^{\prime\prime}=\mu(\tilde{\xi},\tilde{\eta})
\]
is in Schwartz class in $\hat{M}$. Because of the support of $\chi_{3}$,
\[
v(x,y^{\prime})=(\mathcal{F}^{-1}\chi_{3})(x,y^{\prime})
\]
satisfies the constraint. While $v$ may be nonzero on $M$, as may its
derivatives, it is the case that for homogeneous polynomials $p_{k+1}%
(x^{\prime\prime},y^{\prime\prime})$ of degree $k+1$, the function
$p_{k+1}(x^{\prime\prime},y^{\prime\prime})v(x,y^{\prime})$ on $N$ vanishes on
$M$ to at least order $k$. Furthermore $p_{k+1}v$ satisfies the constraint.
Indeed,%
\[
(\mathcal{F}p_{k+1}v)(\xi,\eta^{\prime})=p_{k+1}(\frac{1}{i}\partial
_{\xi^{\prime\prime}},\frac{1}{i}\partial_{\eta^{\prime\prime}})\chi_{3}%
(\xi,\eta^{\prime})\text{,}%
\]
and differential operators do not affect the support.
Set data $u_{0}(x,y^{\prime})=(p_{k+1}v)(x,y^{\prime})$ and $u_{1}=0$, and
solve equation (\ref{ultra}). Because this data satisfies the constraint, the
solution $u(x,y)$ is global. Because of the properties of the initial data,
all $x$ and $y^{\prime}$ derivatives of $u(x,y)$ vanish on $M$. Because
$u_{1}=0$ and $u$ itself satisfies equation (\ref{ultra}), all $y_{1}$
derivatives up to order $k$ as well as any mixed derivatives also vanish.
\end{proof}

\section{A variant of a uniqueness theorem of Courant}

Courant (1962) gives a uniqueness result for the ultrahyperbolic equation with
data posed on a hypersurface of mixed signature, which in our notation states
that, among $C^{2}$ solutions, initial values of $u(x,0,y^{\prime})$ and
$\partial_{y_{1}}u(x,0,y^{\prime})$ prescribed in the set in the Cauchy
hypersurface $M$ given by
\begin{equation}
\sum_{\ell=1}^{d_{1}}(x_{\ell}-x_{\ell}^{0})^{2}\leq a^{2}~,\qquad\sum
_{\ell=2}^{d_{2}}(y_{\ell}-y_{\ell}^{0})^{2}\leq\varepsilon^{2}%
\end{equation}
will determine \textit{a priori} the values of the data on the larger set
\begin{equation}
\left\{  (x,y^{\prime})\in M\ :\ \sqrt{\sum_{\ell=1}^{d_{1}}(x_{\ell}-x_{\ell
}^{0})^{2}}+\sqrt{\sum_{\ell=2}^{d_{2}}(y_{\ell}-y_{\ell}^{0})^{2}}\leq
a\right\}  ~.\label{Eqn:CourantStatement1}%
\end{equation}
Furthermore the solution is determined uniquely in the space-time region
\begin{equation}
\left\{  (x,y)\in{\mathbb{R}}^{d_{1}+d_{2}}\ :\ \sqrt{\sum_{\ell=1}^{d_{1}%
}(x_{\ell}-x_{\ell}^{0})^{2}}+\sqrt{\sum_{\ell=1}^{d_{2}}(y_{\ell}-y_{\ell
}^{0})^{2}}\leq a\right\}  ~.\label{Eqn:CourantStatement2}%
\end{equation}
Courant's proof of this fact uses the Asgeirsson mean value theorem in a
fundamental way.

The key implication from our point of view is that data on an arbitrarily
small cylindrical subset of $M$, plus the stipulation of $C^{2}$ regularity,
determine the data and indeed the solution on much larger sets of $M$ and of
space-time, respectively. In turn, knowledge of the data in a small cylinder
determines the values of all of its derivatives on $N=\{(x,y)\ :\ y=0\}$ (if
the data are smooth). This contrasts to the case discussed in
section~\ref{Sec:FailureUniqueness}, in which it is shown that specification
of a possibly large but finite number of derivatives does not lead to unique
solutions, even when the constraint is imposed and the resulting solutions are
globally defined and smooth.

In this section we give a version of the above theorem of Courant, for data
posed in ellipsoidal domains in the Cauchy hypersurface $M$, which are
localized near the $\{y^{\prime}=0\}$ coordinate axis (or any translate
thereof). Our proof of this result is based on the Holmgren--John theorem
(John 1982), and therefore remains true under perturbations to the equation.
Thus it is a robust generalization of the Courant result, which being based on
Asgeirsson's theorem is true only for precisely the ultrahyperbolic equation.

\begin{theorem}
\label{thm11}
Let $\varepsilon >0$ and define the ellipsoid
$Z_\varepsilon \subseteq M$ by
\begin{equation}
Z_\varepsilon = \{(x,y) \ : \ y_1 = 0 \ ,
|x|^2 +\frac{|y'|^2}{\varepsilon^2} < 1 \} ~, \qquad
0 < \varepsilon \leq 1~ .
\end{equation}
A $C^2$ solution to \eqref{ultra} whose Cauchy data vanishes on
$Z_\varepsilon$ must necessarily vanish on the set
\[
D = \{ (x,y) \in {\mathbb R}^{d_1+d_2} \ : \ |x| + |y| < 1 \} ~
\]
and in particular its Cauchy data along with all derivatives
must vanish on the subset  of the Cauchy hypersurface given by
$\{ (x,y') \in M \ : \ |x| + |y'| < 1 \}$.
\end{theorem}

\begin{proof}
Define $R_\varepsilon(w)$ to be the cone over $Z_\varepsilon$ with
vertex $v = (0,w_1,w') \in \{(x,y) \ : \ x=0 \}$. We will show
that for any $w = (w_1,w')$ with $|w| \leq 1$ (namely the unit sphere
in ${\mathbb R}^{d_2}$), the region between the cone
$R_\varepsilon(w)$ and the ellipsoid $Z_\varepsilon$ is a region
of determinacy for the ultrahyperbolic equation. The closure of the
envelope of such ellipsoidal cones includes the region $D$; in fact it
is slightly larger. The result will follow accordingly.

For a given $R_\varepsilon(w)$, the Holmgren--John theorem is based
upon the construction of an analytic family of noncharacteristic
hypersurfaces $S_\lambda$ with which to sweep the region between
$Z_\varepsilon$ and $R_\varepsilon(w)$. Taking the case of the vertex
$v = (0,w)$ with $w=e_1:=(1,0)$, define
\[
S_\lambda := \{ (x,y) \ : \ (1-y_1)^2 - \bigl( |x|^2 +
\frac{|y'|^2}{\varepsilon^2}\bigr) = - \lambda \}
\]
with $-1 \leq \lambda \leq 0$. The normal to $S_\lambda$ is
$N_\lambda = -2(x,(1-y_1),y'/\varepsilon^2)^T$, so that the
characteristic form calculated on $N_\lambda$ is
\[
\frac{1}{4} N_\lambda^T
\begin{pmatrix} -I_{d_1\times d_1} & 0 \\
0 & I_{d_2\times d_2}
\end{pmatrix} N_\lambda
= \frac{1}{4} \bigl( -|x|^2 + (1-y_1)^2 + |y'|^2/\varepsilon^2
\bigr) ~.
\]
Taking into account that $(x,y_1,y') \in S_\lambda$ and solving
for $(1-y_1)^2$,
\[
\frac{1}{4} N_\lambda^T
\begin{pmatrix} -I_{d_1\times d_1} & 0 \\
0 & I_{d_2\times d_2}
\end{pmatrix} N_\lambda
= \frac{1}{4} \Bigl( \bigl(\frac{1+\varepsilon^2}{\varepsilon^4}\bigr)
|y'|^2 - \lambda \Bigr) ~.
\]
Recalling that $\lambda < 0$ (except in the limiting case
$S_\lambda \to R_\varepsilon$) observe that this family of hyperboloids constitutes
a noncharacteristic analytic family which sweeps the region between
$Z_\varepsilon$ and $R_\varepsilon(e_1)$. Thus the Holmgren--John
uniqueness theorem applies, and this region is a region of determinacy
for the ultrahyperbolic equation \eqref{ultra}.

We have already achieved the analogue of the statement
\eqref{Eqn:CourantStatement2} of Courant. Namely, given
the values of a $C^2$ solution $u(x,y)$
to \eqref{ultra} in the space-time ellipsoid
\[
W_\varepsilon := \{(x,y) \ : \
|x|^2 +\frac{|y|^2}{\varepsilon^2} < 1 \} ~,
\]
we may slice it with a hyperplane which contains the $x$-coordinate
axes but which is otherwise arbitrarily oriented in $y$, to determine
a possible $Z_\varepsilon$,
which in turn determines the solution over the larger conical region
$R_\varepsilon$ with base $Z_\varepsilon$. All of these regions have
been shown to be domains of determinacy. Their union contains the
set $D = \{ (x,y) \ : \ |x| + |y| < 1 \}$. Therefore if a
solution vanishes in $W_\varepsilon$ it must also vanish in $D$.

Returning to the problem of the domain of determinacy of the set
$Z_\varepsilon \subseteq M$, we
generalize the above construction to any $w \in {\mathbb R}^{d_2}$ with
$|w| = 1$. Let $w = Re_1$ for $e_1 = (1,0\dots)$, where $R$ is an
orthogonal matrix. Changing variables to $z = Ry$
and using a symmetric matrix $Q$ of signature $(-,+ \dots)$, an
analytic family of hyperboloids is given by
\[
S_\lambda(w) := \{ (x,z) \ : \
|x|^2 + \langle (z-e_1), Q (z-e_1) \rangle = \lambda \}~,
\]
where the Euclidean inner product is given by $\langle \cdot,
\cdot\rangle$.
The matrix $Q$ is to be chosen so that the intersections of
the hyperboloids $S_\lambda(w)$ with the hypersurface $M$ lie
in $Z_\varepsilon$, and sweeps it as $\lambda$ is varied.

At this point we may assume without loss of generality that
$w=(w_1,w') = (w_1,w_2, 0 \dots)$, whereupon $Q$
may be chosen such that
\[
Q = \begin{pmatrix}  Q_2 & 0 \\
0 & \frac{1}{\varepsilon^2} I'' \end{pmatrix}
~,
\qquad Q_2^T = Q_2 ~,
\]
for $Q_2$ a $2\times 2$ symmetric matrix with signature $(-,+)$.
Furthermore, the above rotation is then set to be
\[
R = \begin{pmatrix}  R_2 & 0 \\
0 &   I'' \end{pmatrix} ~, \qquad
R_2 = \begin{pmatrix} \cos(\theta) & \sin(\theta) \\
-\sin(\theta) & \cos(\theta) \end{pmatrix} ~.
\]
In $y-$coordinates the hyperboloid family is expressed
\[
S_\lambda(w) := \{ (x,y) \ : \
|x|^2 + \langle (y-w), R^T Q R(y-w) \rangle = \lambda \}~,
\]
and the stipulation is that $S_0(w)$ should intersect the hypersurface
$M$ in the original ellipsoid $Z_\varepsilon$. This imposes the
condition that
\[
|x|^2 + \langle (x,0,y'), R^T Q R(x,0,y') \rangle
:= |x|^2 + \langle (x,0,y'), B (x,0,y') \rangle
= |x|^2 + \frac{1}{\varepsilon^2}|y'|^2 ~,
\]
where $B_2$ is the upper left-hand $2\times 2$ block of the matrix
$B$. Therefore one finds the matrix elements of $B_2$
\[
b_{11} = -\frac{\varepsilon^2 - \sin^2(\theta)}{\varepsilon
\cos^2(\theta)} ~, \quad
b_{12} = -\frac{\tan(\theta)}{\varepsilon^2} ~, \quad
b_{22} = \frac{1}{\varepsilon^2} ~,
\]
and furthermore, the $2\times 2$ matrix $Q_2$ is
\begin{equation}\label{Eqn:Q2}
Q_2 = \begin{pmatrix} -1 & \tan(\theta) \\
\tan(\theta) & \frac{1}{\varepsilon^2}a \end{pmatrix} ~,
\end{equation}
where $a = a(\varepsilon,\theta) = (1 + (1-\varepsilon^2)\tan^2(\theta))$.
Calculating the characteristic form on the hyperboloids $S_\lambda(w)$,
we compute the normal $N_\lambda(w)$ as
\[
-\frac{1}{2} N_\lambda(w) = (x, Q(z-e_1))^T ~.
\]
Noting that the characteristic form is invariant under rotations
$R$ as above, which leave the coordinate subspaces
${\mathbb R}^{d_1}_x$ and ${\mathbb R}^{d_2}_y$ invariant,
we find that
\[
\frac{1}{4} N_\lambda(w)^T \begin{pmatrix} -I_{d_1\times d_1} & 0 \\
0 & I_{d_2\times d_2} \end{pmatrix} N_\lambda(w)
= -|x|^2 + \langle (z-e_1), Q^2 (z-e_1)\rangle ~.
\]
This is evaluated on the hyperboloid $S_\lambda(w)$, on which
\[
|x|^2 + \langle (z-e_1), Q(z-e_1)\rangle = \lambda ~.
\]
Solving for $|x|^2$, we find
\[
\frac{1}{4} N_\lambda(w)^T \begin{pmatrix} -I_{d_1\times d_1} & 0 \\
0 & I_{d_2\times d_2} \end{pmatrix} N_\lambda(w)
= \langle (z-e_1), [Q^2 + Q](z-e_1)\rangle - \lambda ~.
\]
Specifically, the matrix $[Q^2 + Q]$ is
\[
[Q^2 + Q] = \begin{pmatrix} Q_2^2 + Q_2 & 0 \\
0 & \bigl( \frac{1+\varepsilon^2}{\varepsilon^4} \bigr)
I'' \end{pmatrix} ~.
\]
Using the form \eqref{Eqn:Q2} for $Q_2$, one calculates
\[
[Q_2^2 + Q_2] = \begin{pmatrix}
\tan^2(\theta) & \frac{a}{\varepsilon^2}\tan(\theta) \\
\frac{a}{\varepsilon^2}\tan(\theta) &
\frac{a^2}{\varepsilon^4} + \tan^2(\theta)
\end{pmatrix} ~.
\]
It is easily verified that this is positive definite. Recalling that
$\lambda \leq 0$ in the definition of the analytic families of
hyperboloids, it follows that $S_\lambda(w)$ are all noncharacteristic,
and hence the Holmgren--John theorem applies, thus completing the
argument.
\end{proof}

\begin{acknowledgements}
\textbf{Acknowledgements:}
The research of the first author was partially supported by the Canada Research
Chairs Program and NSERC grant \#238452-06.  The research of the
second author was partially supported by a grant from SSHRC.
\end{acknowledgements}

\end{document}